\documentclass[journal, 10pt]{IEEEtran}
\usepackage{amsfonts}%
\usepackage{amssymb}%
\usepackage{setspace}
\usepackage[dvips]{graphicx}%
\usepackage{amsmath}%
\usepackage[usenames]{color}
\usepackage{algorithmic}
\usepackage{cite}
\usepackage{subfig}

\usepackage{array}
\newtheorem{theorem}{Theorem}
\newtheorem{corollary}{Corollary}
\newtheorem{lemma}{Lemma}
\newtheorem{assumption}{Assumption}

\newtheorem{example-non*}{Example}

\newtheorem{remark}{Remark}
\allowdisplaybreaks 

\DeclareMathOperator*{\argmax}{arg\,max}

\ifodd 0
\newcommand{\delete}[1]{}
\else
\newcommand{\delete}[1]{#1}
\fi

\newcommand{\comment}[1]{}

\ifodd 1
\newcommand{\remove}[1]{}
\newcommand{\add}[1]{#1}
\else
\newcommand{\remove}[1]{#1}
\newcommand{\add}[1]{}
\fi

\ifodd 1
\newcommand{\bremove}[1]{}
\newcommand{\badd}[1]{#1}
\else
\newcommand{\bremove}[1]{#1}
\newcommand{\badd}[1]{}
\fi

\ifodd 1
\newcommand{\rmv}[1]{}
\else
\newcommand{\rmv}[1]{{\color{red}#1}}
\fi

\ifodd 0
\newcommand{\newc}[1]{{\color{blue}#1}} 
\else
\newcommand{\newc}[1]{#1}
\fi

\ifodd 0
\newcommand{\rev}[1]{{\color{blue}#1}} 
\newcommand{\clar}[1]{\textbf{\color{green}(NEED CLARIFICATION: #1)}}

\else
\newcommand{\rev}[1]{#1}

\newcommand{\clar}[1]{}
\fi

\ifodd 0

\newcommand{\com}[1]{\textbf{\color{red}(COMMENT: #1)}} 

\else
\newcommand{\com}[1]{}
\fi

\begin{document}
\title{Distributed Online Learning in Social Recommender Systems}
\author{\IEEEauthorblockN{Cem Tekin,~\IEEEmembership{Member,~IEEE}, Simpson Zhang, Mihaela van der Schaar,~\IEEEmembership{Fellow,~IEEE}\\}
\thanks{This work is partially supported by the grants NSF CNS 1016081 and AFOSR DDDAS.}
}
\maketitle

\begin{abstract}
In this paper, we consider decentralized sequential decision making in distributed online recommender systems, where items are recommended to users based on their search query as well as their specific background including history of bought items, gender and age, all of which comprise the context information of the user.
In contrast to centralized recommender systems, in which there is a single centralized seller who has access to the complete inventory of items as well as the complete record of sales and user information, in decentralized recommender systems each seller/learner only has access to the inventory of items and user information for its own products and not the products and user information of other sellers, but can get commission if it sells an item of another seller. Therefore the sellers must distributedly find out for an incoming user which items to recommend (from the set of own items or items of another seller), in order to maximize the revenue from own sales and commissions.
We formulate this problem as a cooperative contextual bandit problem, analytically bound the performance of the sellers compared to the best recommendation strategy given the complete realization of user arrivals and the inventory of items, as well as the context-dependent purchase probabilities of each item, and verify our results via numerical examples on a distributed data set adapted based on Amazon data.
%
We evaluate the dependence of the performance of a seller on the inventory of items the seller has, the number of connections it has with the other sellers, and the commissions which the seller gets by selling items of other sellers to its users.
%
\end{abstract}


\vspace{-0.1in}
\section{Introduction}\label{sec:intro}

One of the most powerful benefits of a social network is the ability for cooperation and coordination on a large scale over a wide range of different agents \cite{chen2012technological}. By forming a network, agents are able to share information and opportunities in a mutually beneficial fashion. For example, companies can collaborate to sell products, charities can work together to raise money, and a group of workers can help each other search for jobs. Through such cooperation, agents are able to attain much greater rewards than would be possible individually. But sustaining efficient cooperation can also prove extremely challenging. First, agents operate with only incomplete information, and must learn the environment parameters slowly over time. Second, agents are decentralized and thus uncertain about their neighbor's information and preferences. \newc{Finally, agents are selfish in the sense that, they don't want to reveal their inventory, observations and actions to other agents, unless they benefit.
This paper produces a class of algorithms that effectively addresses all of these issues: at once allowing decentralized agents to take near-optimal actions in the face of incomplete information, while still incentivizing them to fully cooperate within the network.}

The framework we consider is very broad and applicable to a wide range of social networking situations. We analyze a group of agents that are connected together via a fixed network, each of whom experiences inflows of users to its page. Each time a user arrives, an agent chooses from among a set of items to offer to that user, and the user will either reject or accept each item. These items can represent a variety of things, from a good that the agent is trying to sell, to a cause that the agent is trying to promote, to a photo that the agent is trying to circulate. In each application, the action of accepting or rejecting by the user will likewise have a distinct meaning. When choosing among the items to offer, the agent is uncertain about the user's acceptance probability of each item, but the agent is able to observe specific background information about the user, such as the user's gender, location, age, etc. Users with different backgrounds will have different probabilities of accepting each item, and so the agent must learn this probability over time by making different offers.

We allow for cooperation in this network by letting each agent recommend items of neighboring agents to incoming users, in addition to its own items. Thus if the specific background of the incoming user makes it unlikely for him to accept any of the agent's items, the agent can instead recommend him some items from a neighboring agent with more attractive offerings. By trading referrals in this fashion, all of the agents that are connected together can benefit. To provide proper incentives, a commission will be paid to the recommending agent every time an item is accepted by a user from the recommended agent. When defined appropriately, this commission ensures that both sides will benefit each time a recommendation occurs and thus is able to sustain cooperation.

However, since agents are decentralized, they do not directly share the information that they learn over time about user preferences for their own items.  So when the decision to recommend a neighboring agent occurs, it is done based solely on the previous successes the agent had when recommending that neighbor. Thus agents must learn about their neighbor's acceptance probabilities through their own trial and error, unlike in other social learning papers such as \cite{krishnamurthy2012quickest, krishnamurthy2013social, jadbabaie2012non, nedic2009distributed}, where agents share information directly with their neighbors. 

Another key feature of our algorithms is that they are non-Bayesian unlike \cite{krishnamurthy2012quickest, krishnamurthy2013social}. Instead we model the learning through contextual bandits, where the context is based on the user's background. By building upon the theory of contextual bandits, we produce a class of algorithms that allows agents to take near-optimal actions even with decentralized learning. 
\newc{We prove specific bounds for the regret, which is the difference between the total expected reward of an agent using a learning algorithm and the total expected reward of the optimal policy for the agent, which is computed given perfect knowledge about acceptance probabilities for each context. We show that the regret is sublinear in time in all cases.}
\newc{We further show that our algorithms can operate regardless of the specific network topology, including the degree of connectivity, degree distribution, clustering coefficient, etc., although the performance is better if the network is more connected since each agent will have access to more items of other agents.} 

The rest of the paper is organized as follows.
%
Related work is given in Section \ref{sec:related}. The problem formulation is given in Section \ref{sec:probform}. In Section \ref{sec:mainalgo}, we consider the online learning problem involving multiple decentralized agents and analyze its regret. In Section \ref{sec:contextdep} we develop an algorithm to achieve sublinear regret when the purchase probability of the items depend on the other recommended items, and in Section \ref{sec:contextindep} we develop a faster learning algorithm when item purchase probabilities are independent of each other. 
The effect of connectivity between the agents is given in Section \ref{sec:discuss}. In Section \ref{sec:num}, numerical results demonstrating the effects of commissions, size of the set of items of agents and connectivity of agents are given, using an artificial data set which is based on a real data set. Finally, we conclude the paper in Section \ref{sec:conc}.

\section{Related Work}\label{sec:related}

This work represents a significant departure from the other works in contextual bandits, which consider only centralized agents, single arms played at once, and no incentive issues. Most of the prior work on contextual bandits is focused on a single agent choosing one arm at a time based on the context information provided to it at each time slot \cite{slivkins2009contextual, dudik2011efficient, langford2007epoch, chu2011contextual}.
In these works the system is centralized, so the agent can directly access all the arms. Our framework in this paper differs from the centralized contextual bandit framework in two important aspects. First, multiple agents who can only access a subset of arms, and who only get feedback about this subset, cooperate to maximize their total reward. Second, each agent can choose multiple arms at each time slot, which makes the arm selection problem combinatorial in the number of arms.
To the best of our knowledge, our work is the first to provide rigorous solutions for online learning by multiple cooperative agents selecting multiple arms at each time step when context information is present. We had previously proposed a multi-agent contextual bandit framework in \cite{tekin2013bigdata} where each agent only selects a single arm at each time slot. Different from this work, in this paper we assume that an agent can select multiple arms, and the expected reward of the agent from an arm may depend on the other selected arms. This makes the problem size grow combinatorially in the arm space, which requires the design of novel learning algorithms to quicken the learning process. 
Combinatorial bandits \cite{gai2012combinatorial} have been studied before in the multi-armed bandit setting, but to the best of our knowledge we are the first to propose the 
decentralized contextual combinatorial bandit model studied in this paper. 
This decentralization is important because it allows us to analyze a social network framework and the fundamental challenges associated with it including commissions, third-party sellers, etc. We are also able to consider the specific effects of the network structure on the regret in our model. 
In contrast, our approach in \cite{tekin2013bigdata} does not address the network structure concerns.
Several other examples of related work in contextual bandits are \cite{li2010contextual}, in which a contextual bandit model is used for recommending personalized news articles based on a variant of the UCB1 algorithm in \cite{auer} designed for linear rewards, and \cite{crammer2011multiclass} in which  the authors solve a classification problem using contextual bandits, where they propose a perceptron based algorithm that achieves sublinear regret.

%
\newc{Apart from contextual bandits, there is a large set of literature concerned in multi-user learning using a multi-armed bandit framework \cite{anandkumar, liu2010distributed, tekin2012online, 6362216, tekin2012sequencing}.}
%
%
%
%
%
We provide a detailed comparison between our work and related work in multi-armed bandit learning in Table \ref{tab:comparison2}. Our cooperative contextual learning framework can be seen as an important extension of the centralized contextual bandits framework \cite{slivkins2009contextual}. The main differences are that: (i) a three phase learning algorithm with {\em training}, {\em exploration} and {\em exploitation} phases is needed instead of the standard two phase algorithms with {\em exploration} and {\em exploitation} phases that are commonly used in centralized contextual bandit problems; (ii) the adaptive partitions of the context space should be formed in a way that each learner can efficiently utilize what is learned by other learners about the same context; (iii) since each agent has multiple selections at each time slot, the set of actions for each agent is very large, making the learning rate slow. Therefore, the correlation between the arms of the agents should be exploited to quicken the learning process.
In our distributed multi-agent multiple-play contextual bandit formulation, the training phase, which balances the learning rates of the agents, is necessary since the context arrivals to agents are different which makes the learning rates of the agents for various context different.

\begin{table}[t]
\centering
{\fontsize{10}{10}\selectfont
\setlength{\tabcolsep}{.25em}
\begin{tabular}{|l|c|c|c|c|c|}
\hline
&\cite{slivkins2009contextual, dudik2011efficient, langford2007epoch, chu2011contextual} &  \cite{6362216, anandkumar}   &  \cite{tekin2013bigdata} & This work \\
& & \cite{tekin2012sequencing} & & \\
\hline
Multi-agent & no & yes  & yes & yes \\
\hline
Cooperative & N/A & yes  & yes & yes \\
\hline
Contextual & yes & no  & yes & yes\\
\hline
Context arrival  & arbitrary & N/A  & arbitrary & arbitraty\\
process & & & & \\
\hline
(syn)chronous,& N/A & syn & both & both \\
(asyn)chronous& & & & \\
\hline
Regret & sublinear & log  & sublinear & sublinear\\
\hline
Multi-play for  & no & no & no & yes \\
each agent & & & & \\
\hline
Action set size  & no & no & no & yes \\
combinatorial in & & &  &  \\
number of agents &  &  &  &  \\
and items &  & &  &  \\
\hline
Action sets of  & N/A & same & different & different \\
the agents & & & & \\
\hline
\end{tabular}
}
\caption{Comparison with related work in multi-armed bandits.}
\vspace{-0.35in}
\label{tab:comparison2}
\end{table}

There is also an extensive literature on recommender systems that incorporates a variety of different methods and frameworks.  Table \ref{tab:reccomparison} provides a summary of how our work is related to other work. Of note, there are several papers that also use a similar multi-armed bandit framework for recommendations.
For example, \cite{ deshpande2012linear } considers a bandit framework where a recommender system learns about the preferences of users over time as each user submits ratings. It uses a linear bandit model for the ratings, which are assumed to be functions of the specific user as well as the features of the product.
%
\cite{kohli2013fast} is another work that utilizes multi-armed bandits in a recommendation system. It considers a model that must constantly update recommendations as both preferences and the item set changes over time.
%
	
There are also numerous examples of works that do not use a bandit framework for recommendations. One of the most commonly used methods for recommendations are collaborative filtering algorithms such as  \cite{hofmann2004latent, sahoo2012hidden, linden2003amazon, panniello2012comparing, miyahara2000collaborative, o1999clustering, shani2002mdp}, which make recommendations by predicting the user's preferences based on a similarity measure with other users. Items with the highest similarity score are then recommended to each user; for instance items may be ranked based on the number of purchases by similar users. There are numerous ways to perform the similarity groupings, such as the cluster model in \cite{linden2003amazon, o1999clustering} that groups users together with a set of like-minded users and then makes recommendations based on what the users in this set choose. Another possibility is presented in \cite{panniello2012comparing}, which pre-filters ratings based on context before the recommendation algorithm is started. 
	
	
An important difference to keep in mind is that the recommendation systems in other works are a single centralized system, such as Amazon or Netflix. Thus the system has complete information at every moment in time, and does not need to worry about incentive issues or pay commissions. However, in this paper each agent is in effect its own separate recommendation system, since agents do not directly share information with each other. Therefore the mechanism we propose must be applied separately by every agent in the system based on that agent's history of user acceptances. So in effect our model is a giant collection of recommendation systems that are running in parallel. The only cross-over between these different systems is when one agent suggests which of its items should be recommended by another agent. This allows for an indirect transfer of information, and lets that other agent make better choices than it could without this suggested list of items. 

Also, it is important to note that decentralization in the context of our paper does not mean the same thing as in other papers such as \cite{ziegler2005towards}. Those papers assume that the $users$ are decentralized, and develop mechanisms based on trust and the similarity of different users in order to allow $users$ to recommend items to each other. We are assuming that the $agents$ are decentralized, and so each user still only gets a recommendation from one source, it is just that this source may be different depending on which agent this user arrives at. Thus this paper is fundamentally different from the works in that literature.

\begin{table}[t]
\centering
{\fontsize{8}{7}\selectfont
\setlength{\tabcolsep}{.25em}
\begin{tabular}{|l|c|c|c|c|c|c|}
\hline
   & Item-  & Memory- & Uses  & Performan- & Similarity & Central- \\
   & based  & based, & context & ce  & distance & ized(C),   \\
   & (IB), & model- & info. & measure & & Decent- \\
   & user- & based & &  & & ralized(D) \\
   & based & & & & & \\
   & (UB) & & & & & \\
\hline
\cite{ balabanovic1997fab }  & UB & Memory- & No & Ranking  & - & C \\
& & based & & precision & & \\
\hline
\cite{ hofmann2004latent }  & UB & Bayesian- & No & MAE,  & Pearson  & C \\
  & & based latent & & RMS,  & correlation & \\
  & & semantic  & & 0/1 loss & & \\
  & & model & & & & \\
\hline
\cite{ sahoo2012hidden } & UB & Bayesian- & No & Precision\& & Pearson  & C\\
  & & based  & & Recall & correlation & \\
  & & Markov & & & & \\
  & & model & & & & \\
  
\hline
\cite{linden2003amazon} & IB & Cluster model & No & - & Cosine & C \\
\hline
\cite{panniello2012comparing}  & UB & Memory- & Yes & Precision\&  & - & C\\
& & based & & Recall & & \\
\hline
\cite{ miyahara2000collaborative }  & UB & Bayesian  & No & Precision\& & Pearson  & C\\
&  & classifier & & Recall & correlation & \\
& & model & & & & \\
\hline
\cite{o1999clustering}  & UB & Cluster model & No & MAE\& & Pearson  & C\\
& & & & Coverage  & correlation & \\
\hline
\cite{shani2002mdp}  & UB & MDP model & No & Recall & Self-defined  & C\\
& & & & & similarity & \\
\hline
\cite{deshpande2012linear}  & UB & MAB model & No & Reward & Lipschitz   & C \\
& & & & & continuous & \\
\hline
\cite{kohli2013fast}  & UB & MAB model & Yes & Regret & Lipschitz & C \\
& & & & & continuous & \\
\hline
Our  & UB & MAB model & Yes & Regret & H\"{o}lder  & D \\
work & & & & & continuous & \\
\hline
\end{tabular}
}
\caption{Comparison with prior work in recommender systems.}
\vspace{-0.2in}
\label{tab:reccomparison}
\end{table}

\vspace{-0.2in}
\section{Problem Formulation}\label{sec:probform}

There are $M$ decentralized agents/learners which are indexed by the set ${\cal M} := \{1,2,\ldots,M\}$. Let ${\cal M}_{-i} := {\cal M} - \{ i\}$. Each agent $i$ has an inventory of items denoted by ${\cal F}_i$, which it can offer to its users and the users of other agents by paying some commission.
Users arrive to agents in a discrete time setting ($t=1,2,\ldots$).
Agent $i$ recommends $N$ items to its user at each time slot. For example, $N$ can be the number of recommendation slots the agent has on its website, or it can be the number of ads it can place in a magazine. We assume that $N$ is fixed throughout the paper. 
\newc{Each item $f \in {\cal F}_i$ has a fixed price $p^i_{f}>0$. 
We assume that the inventories of agents are mutually disjoint.\footnote{Even when an item is in the inventory of more than one agent, its price can be different among these agents, thus different IDs can be assigned to these items. We do not assume competition between the agents, hence our methods in this paper will work even when the inventories of agents are not mutually disjoint.}} 
\newc{For now, we will assume that all the agents in the network are directly connected to each other, so that any agent can sell items to the users of any other agent directly without invoking intermediate agents. We will discuss the agents in more general network topologies in Section \ref{sec:discuss}.} Let ${\cal F} := \cup_{i \in {\cal M}} {\cal F}_i$ be the set of items of all agents.
We assume that there is an unlimited supply of each type of item. This assumption holds for digital goods such as e-books, movies, videos, songs, photos, etc. 
%
%
An agent does not know the inventory of items of the other agents but knows an upper bound on $|{\cal F}_j|$,\footnote{For a set $A$, $|A|$ denotes its cardinality.} $j \in {\cal M}_{-i}$ which is equal to $F_{\max}$. 

We note that our model is suitable for a broad range of applications. The agent can represent a company, an entrepreneur, a content provider, a blogger, etc., and the items can be goods, services, jobs, videos, songs, etc. The notion of an item can be generalized even further to include such things as celebrities that are recommended to be followed in Twitter, Google+, etc. And the prices that we introduce later can also be generalized to mean any type of benefit the agent can receive from a user. For expositional convenience, we will adopt the formulation of firms selling goods to customers for most of the paper, but we emphasize that many other interpretations are equally valid and applicable.
\vspace{-0.1in}
\subsection{Notation} \label{sec:notation}

Natural logarithm is denoted by $\log(\cdot)$.
For sets ${\cal A}$ and ${\cal B}$, ${\cal A} - {\cal B}$ denotes the elements of ${\cal A}$ that are not in  ${\cal B}$.
$P(\cdot)$ is the probability operator, $E[\cdot]$ is the expectation operator.
For an algorithm $\alpha$, $E_{\alpha}[\cdot]$ denotes the expectation with respect to the distribution induced by $\alpha$.
Let $[t] := \{1,\ldots,t-1\}$. Let $\beta_2 := \sum_{t=1}^\infty 1/t^2$.

\vspace{-0.1in}
\subsection{Definition of users with contexts} \label{sec:regulator}

At each time slot $t=1,2,\ldots$, a user with a specific search query indicating the type of item the user wants, or other information including a list of previous purchases, price-range, age, gender etc., arrives to agent $i$.
We define all the properties of the arriving user known to agent $i$ at time $t$ as the context of that user, and denote it by $x_i(t)$. We assume that the contexts of all users belong to a known space ${\cal X}$, which without loss of generality is taken to be $[0,1]^d$ in this paper, where $d$ is the dimension of the context space.
%
%
Although the model we propose in this paper has synchronous arrivals, it can be easily extended to the asynchronous case where agents have different user arrival rates, and even when no user arrives in some time slots. 
The only difference of this from our framework is that instead of keeping the same time index $t$ for every agent, we will have different time indices for each agent depending on the number of user arrivals to that agent. 
%
%

\vspace{-0.1in}
\subsection{Definition of commissions}\label{sec:regulator}

In order to incentivize the agents to recommend each other's items, they will provide commissions to each other. 
%
%
%
\newc{In this paper we focus on {\em sales commission}, which is paid by the recommended agent to the recommending agent every time a user from the recommending agent purchases an item of the recommended agent. We assume that these commissions are fixed at the beginning and do not change over time.}
\rev{The system model is shown in Fig. \ref{fig:system}.} 
Basically, if agent $i$ recommends an item $f$ of agent $j$ to its user, and if that user buys the item of agent $j$, then agent $i$ obtains a fixed commission which is equal to $c_{i,j}>0$. All of our results in this paper will also hold for the case when the commission is a function of the price of the item $f$ sold by agent $j$, i.e., $c_{i,j}(p^j_{f})$. However we use the fixed commission assumption, since privacy concerns may cause the agent $j$ or the user to not want to reveal the exact sales price to agent $i$. We assume that $c_{i,j} \leq p^j_f$ for all $i,j \in {\cal M}$, $f \in {\cal F}_j$.
Otherwise, full cooperation can result in agent $j$ obtaining less revenue than it would have obtained without cooperation.
Nevertheless, even when $c_{i,j} > p^j_f$ for some $f \in \bar{{\cal F}}_j \subset {\cal F}_j$, by restricting the set of items that agent $j$ can recommend to agent $i$ to ${\cal F}_j-\bar{{\cal F}}_j$, all our analysis for agent $i$ in this paper will hold.
%

Agents may also want to preserve the privacy of the items offered to the users, hence the privacy of their inventory in addition to the privacy of prices. 
This can be done with the addition of a {\em regulator} to the system. 
The regulator is a non-strategic entity whose goal is to promote privacy among the agents. If agent $i$ wants to recommend an item from agent $j$, it can pass this request to the regulator and the regulator can create private keys for agent $j$ and the user of agent $i$, such that agent $j$ sends its recommendation to agent $i$'s website in an encrypted format so that recommendation can only be viewed by the user, who has the key. The regulator can also control the transaction between agent $j$ and agent $i$'s user such that agent $j$ pays its commission to agent $i$ if the item is sold. Note that the regulator does not need to store the previous purchase histories or the inventories of the agents. Also it only regulates transactions between the agents but not the recommendations of an agent to its own users.
%
%
%

\begin{figure}
\begin{center}
\includegraphics[width=0.93\columnwidth]{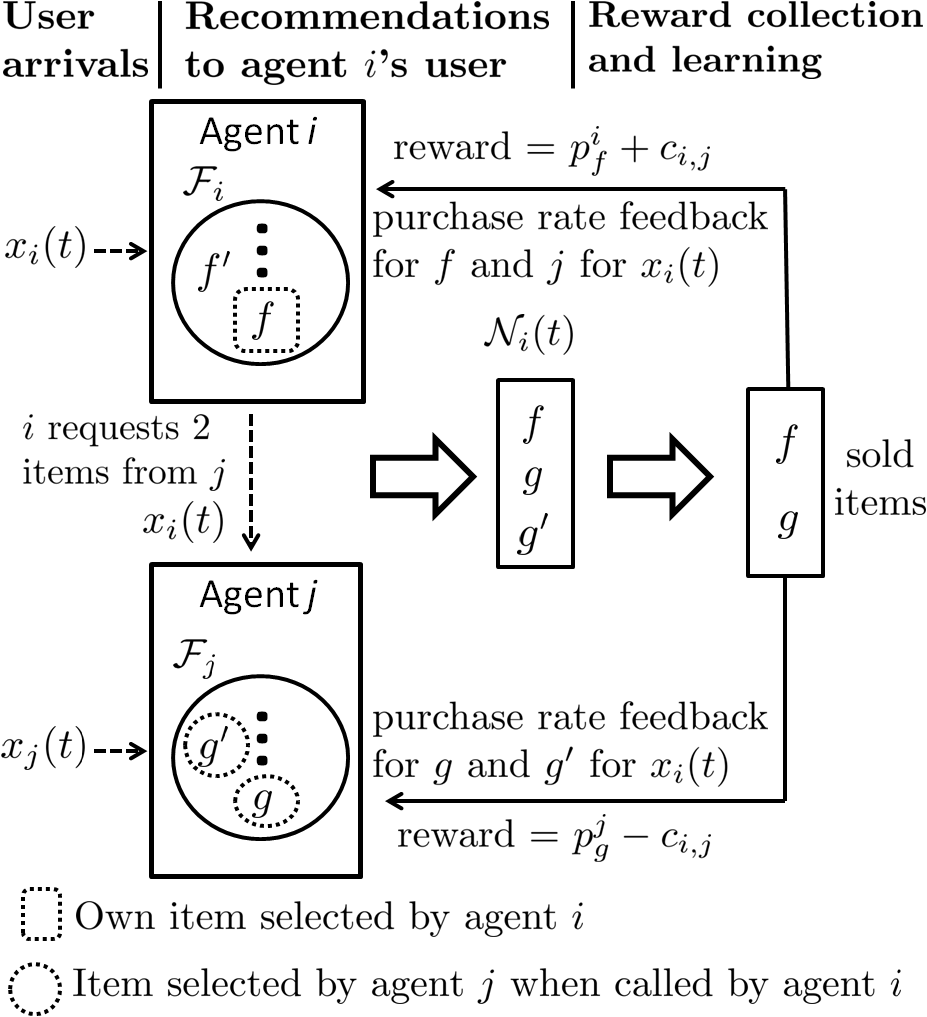}
\caption{Operation of the system for agent $i$ for $N=3$ recommendations. At each time a user arrives to agent $i$ with context $x_i(t)$, agent $i$ recommends a set of its own items and items from other agents, which is denoted by ${\cal N}_i(t)$.} 
\vspace{-0.2in}
\label{fig:system}
\end{center}
\end{figure}

\vspace{-0.1in}
\subsection{Recommendations, actions and purchase probabilities}

The $N$ items that agent $i$ recommends to its user are chosen in the following way:
An item can be chosen from the inventory ${\cal F}_i$ of agent $i$, or agent $i$ can call another agent $j$ and send the context information of $x_i(t)$ of its user, then agent $j$ returns back an item $f'$ with price $p^j_{f'}$\footnote{Another method to preserve the privacy of prices is the following: Agent $j$'s item will be recommended by agent $i$ without a price tag, and when the user clicks to agent $j$'s item it will be directed to agent $j$'s website where the price will be revealed to the user.} to be recommended to agent $i$ based on $x_i(t)$. Let ${\cal N}_i(t)$ be the set of items recommended by agent $i$ to its user at time $t$.
%
%
%
%
Let ${\cal A}_N$ be the set of subsets of ${\cal F}$ with $N$ items. Let ${\cal N} \in {\cal A}_N$ be a set of $N$ recommended items.

Consider agent $i$.  Let ${\cal U}_i := {\cal M}_{-i} \cup {\cal F}_i$. For $u_i \in {\cal F}_i$, let $(u_i,1)$ ($(u_i,0)$) denote the event that $u_i$ is recommended (not recommended) to agent $i$'s user.
For $u_i \in {\cal M}_{-i}$, let $(u_i,n_{u_i})$ denote the event that agent $i$ requests $n_{u_i}$ (distinct) items to be recommended to its user from agent $u_i$.
Let $(\boldsymbol{u}_i,\boldsymbol{n}_{i}) : = \{(u_i, n_{u_i})\}_{u_i \in {\cal U}_{i}}$ be an action for agent $i$.
%
Based on this, let
\begin{align*}
{\cal L}_i &:= \left\{ (\boldsymbol{u}_i,\boldsymbol{n}_{i}) \textrm{ such that }  n_{u_i}
 \in \times \{0,1\} \textrm{ for } u_i \in {\cal F}_i \textrm{ and } \right. \\ 
  &\left. n_{u_i} \in  \{0,\ldots, N\}  \textrm{ for } u_i \in {\cal M}_{-i} \textrm{ and }  \sum_{u_i \in {\cal U}_i} n_{u_i} = N   \right\},
\end{align*}
be the set of actions available to agent $i$.
%
%
We assume that $|{\cal F}_j| \geq N$ for all $j \in {\cal M}$.  
Let $k_i$ be the components of vector $(\boldsymbol{u}_i,\boldsymbol{n}_{i}) \in {\cal L}_i$ with nonzero $n_{u_i}$.
Since $n_{u_i} = 0$ means the corresponding item will not be recommended for $u_i \in {\cal F}_i$, and the corresponding agent will not recommend any items to agent $i$ for $u_i \in {\cal M}_{-i}$, $k_i$ is enough to completely characterize the action of agent $i$.
Thus, for notational simplicity we denote an action of agent $i$ by $k_i$. With an abuse of notation $k_i \in {\cal L}_i$ means that $(\boldsymbol{u}_i,\boldsymbol{n}_{i})$ corresponding to $k_i$ is in ${\cal L}_i$.
For $u_i \in {\cal U}_i$, we let $n_{u_i}(k_i)$ be the value of $n_{u_i}$ in the vector $(\boldsymbol{u}_i,\boldsymbol{n}_{i}) $ corresponding to $k_i$.
Let ${\cal M}_i(k_i)$ be the set of agents in ${\cal M}_{-i}$ that recommend at least one item to agent $i$, when it chooses action $k_i$. We define $m_j(k_i)$ as the number of items agent $j$ recommends to agent $i$ when agent $i$ chooses action $k_i \in {\cal L}_i$.
Clearly we have $m_j(k_i) = n_j(k_i)$ for $j \in {\cal M}_{-i}$, and $m_i(k_i) = \sum_{u_i \in {\cal F}_i} n_{u_i} (k_i)$.
Let $\boldsymbol{m}(k_i) := (m_1(k_i), \ldots, m_M(k_i))$ be the {\em recommendation vector} given action $k_i$,
\begin{align*}
\boldsymbol{m}_{-j}(k_i) := (m_1(k_i), \ldots, m_{j-1}(k_i), m_{j+1}(k_i), \ldots, m_M(k_i)),
\end{align*}
and ${\cal F}_i(k_i) \subset {\cal F}_i$ be the set of own items recommended by agent $i$ to its user when it takes action $k_i \in {\cal L}_i$.
For a set of recommended items ${\cal N} \in {\cal A}_N$, let ${\cal G}_j({\cal N}) := {\cal N} \cap {\cal F}_j$ and $g_j({\cal N}) := |{\cal G}_j({\cal N})|$.
Let $\boldsymbol{g}_{-i}({\cal N}) := \{g_j({\cal N})\}_{j \in {\cal M}_{-i}}$.
Below, we define the purchase probability of an item $f$ that is recommended to a user with context $x$ along with the set of items ${\cal N} \in {\cal A}_N$.
\begin{assumption}\label{ass:dependent}
\textbf{Group-dependent purchase probability:}
We assume that the inventory of each agent forms a separate group (or type) of items. For example, if agents are firms, then firm $i$ may sell televisions while firm $j$ sells headphones.
For each item $f$ recommended together with the set of items ${\cal N} \in {\cal A}_N$, if $f$ is an item of agent $j$, i.e., it belongs to group ${\cal F}_j$, then its acceptance probability will depend on other recommended items in the same group, i.e., ${\cal G}_j({\cal N})$. The acceptance probability of an item $f \in {\cal F}_j$ may also depend on groups of other items offered together with $f$ and also on their number, i.e., $\boldsymbol{g}_{-j}({\cal N})$, but not on their identities.  
For $f \in {\cal F}_j$, a user with context $x \in {\cal X}$ will buy the item with an unknown probability $q_f(x, {\cal N}) := q_{f}(x, {\cal G}_j({\cal N}), \boldsymbol{g}_{-j}({\cal N}))$.
For all $f \in {\cal F}$, users with contexts {\em similar} to each other have {\em similar} purchase probabilities.
This similarity is characterized by a H\"{o}lder condition that is known by the agents, i.e., there exists $L>0$, $\alpha>0$ such that for all $x,x' \in {\cal X} = [0,1]^d$, we have
\begin{align*}
|q_{f}(x, {\cal N}) - q_{f}(x', {\cal N})| \leq L ||x-x'||^\alpha,
\end{align*}
where $||\cdot||$ denotes the Euclidian norm in $\mathbb{R}^d$.
\end{assumption}

In Assumption \ref{ass:dependent}, we do not require the purchase probabilities to be similar for different sets of recommendations from the same group, i.e., for $f \in {\cal N}$ and $f \in{\cal N}'$ such that ${\cal G}_j({\cal N}) \neq {\cal G}_j({\cal N}')$, the purchase probability of $f$ can be different for context $x \in {\cal X}$. Even though the H\"{o}lder condition can hold with different constants $L_{f}$ and $\alpha_{f}$ for each item $f \in {\cal F}$, taking $L$ to be the largest among $L_{f}$ and $\alpha$ to be the smallest among $\alpha_{f}$ we get the condition in Assumption \ref{ass:dependent}.
For example, the context can be the (normalized) age of the user, and users with ages similar to each other can like similar items.
%
%
We will also consider a simplified version, in which the purchase probability of an item is independent of the set of other items recommended along with that item. 

\begin{assumption}\label{ass:independent}
\textbf{Independent purchase probability:} For each item $f$ offered along with the items in the set ${\cal N} \in {\cal A}_N$, a user with context $x$ will buy the item with an unknown probability $q_f(x, {\cal N}) :=q_f(x)$, independent of the other items in ${\cal N}$, for which there exists $L>0$, $\alpha>0$ such that for all $x,x' \in {\cal X}$, we have
%
$|q_{f}(x) - q_{f}(x')| \leq L ||x-x'||^\alpha$.
%
\end{assumption}

When Assumption \ref{ass:independent} holds, the agents can estimate the purchase probability of an item by using the empirical mean of the number of times the item is purchased by users with similar context information. The purchase probabilities of items can be learned much faster in this case compared with Assumption \ref{ass:dependent}, since the empirical mean will be updated every time the item is recommended.
However, when the purchase probability of an item is group-dependent, the agents need to learn its purchase probability separately for each group the item is in.
%
%
\subsection{Objective, rewards and the regret}

The goal of agent $i$ is to maximize its total expected reward (revenue). However, since it does not know the purchase probability of any item or the inventories of other agents a-priori, it must learn how to make optimal recommendations to its users over time.
In this subsection we define the optimal actions for the agents and the ``oracle" benchmark (optimal) solution the agents' learning algorithms compete with, which is derived from the unknown purchase probabilities and inventories.

The expected reward of agent $i$ at time $t$ from recommending a set of items ${\cal N}_i$ to its user with context $x_i(t) =x$ is given by
\begin{align*}
\sigma_{i}({\cal N}_i, x) := \hspace{-0.1in} \sum_{f \in ({\cal N}_i - {\cal F}_i)} c_{i,j(f)} q_f(x, {\cal N}_i) + \hspace{-0.1in} \sum_{f \in {\cal N}_i \cap  {\cal F}_i} p^i_f q_f(x, {\cal N}_i),
\end{align*}
where $j(f)$ is the agent who owns item $f$.
Based on this, given a user with context $x$, the set of items which maximizes the one-step expected reward of agent $i$ is
\begin{align}
{\cal N}^*_i(x) &:= \argmax_{{\cal N} \in {\cal A}_N} \left( \sigma_{i}({\cal N}, x) \right). \label{eqn:bestonestep} 
\end{align}
Since the inventory of other agents and $q_f(x, {\cal N})$, $x \in {\cal X}$, ${\cal N} \in {\cal A}_N$ are unknown a priori to agent $i$, ${\cal N}^*_i(x)$ is unknown to agent $i$ for all contexts $x \in {\cal X}$.
%
For a {\em recommendation vector} $\boldsymbol{m}=(m_1,\ldots,m_M)$, agent $j$ has ${|{\cal F}_j| \choose m_j}$ actions $k_j \in {\cal L}_j$ for which $m_{j'}(k_j) = m_{j'}$ for all $j' \neq j$. 
Denote this set of actions of agent $j$ by ${\cal B}_j(\boldsymbol{m})$. For an action $k_j \in {\cal B}_j(\boldsymbol{m})$, let 
\begin{align*}
\pi_{j,k_j}(x) := \sum_{f \in {\cal F}_j(k_j)} q_f(x,{\cal F}_j(k_j),\boldsymbol{m}_{-j}(k_j)),
\end{align*}
be the {\em purchase rate} of action $k_j$\footnote{If ${\cal F}_j(k_j) = \emptyset$, then $\pi_{j,k_j}(x)=0$.} for a user with context $x$. 
The {\em purchase rate} is equal to the expected number of items in ${\cal F}_j(k_j)$ sold to a user with context $x$, when the number of items recommended together with ${\cal F}_j(k_j)$ from the other groups is given by the vector $\boldsymbol{m}_{-j}(k_j)$.
For actions of agent $i$ such that $m_j(k_i) >0$, the best set of items agent $j$ can recommend to agent $i$ (which maximizes agent $i$'s expected commission from agent $j$) for a user with context $x$ is ${\cal F}_j(k^*_j(k_i,x))$, where 
\begin{align*}
k^*_j(k_i,x) := \argmax_{k_j \in {\cal B}_j(\boldsymbol{m}(k_i))} \pi_{j,k_j}(x).
\end{align*}
According to the above definitions, the expected reward of action $k_i \in {\cal L}_i$ to agent $i$ for a user with context $x$ is defined as 
\begin{align}
\mu_{i,k_i}(x) &:=  \sum_{f \in {\cal F}_i(k_i)}  p^i_f q_f(x,{\cal F}_i(k_i),\boldsymbol{m}_{-i}(k_i)) \notag  \\
&+ \sum_{j \in {\cal M}_{i}(k_i)}  c_{i,j} \pi_{j, k^*_j(k_i,x)}(x). \label{eqn:rewarddef}
\end{align}
Then, given a user with context $x$, the best action of agent $i$ is 
\begin{align}
k^*_i(x) := \argmax_{k_i \in {\cal L}_i} \mu_{i,k_i}(x). \label{eqn:armoptimal}
\end{align}
%


Note that for both Assumptions \ref{ass:dependent} and \ref{ass:independent}, $\mu_{i,k^*_i(x)}(x) = \sigma_{i}({\cal N}^*_i(x), x)$.
Let $\alpha_i$ be the recommendation strategy adopted by agent $i$ for its own users, i.e., based on its past observations and decisions, agent $i$ chooses a vector $\alpha_i(t) \in {\cal L}_i$ at each time slot. Let $\beta_i$ be the recommendation strategy adopted by agent $i$ when it is called by another agent to recommend its own items. Let $\boldsymbol{\alpha} = (\alpha_1, \ldots, \alpha_M)$ and $\boldsymbol{\beta} = (\beta_1, \ldots, \beta_M)$.
Let $Q^i_{\boldsymbol{\alpha}, \boldsymbol{\beta}}(T)$ be  
the expected total reward of agent $i$  by time $T$ from item sales to its own users and users of other agents and commissions it gets from item sales of other agents to its own users. 
For $f \in {\cal F}_i$, let $Y^{\alpha_i}_f(t)$ be a random variable which is equal to 1 when agent $i$ recommends its own item $f$ to its user at time $t$ (0 otherwise), let $Y^{\beta_i, \alpha_j}_f(t)$ be a random variable which is equal to 1 when agent $i$ recommends its own item $f$ to agent $j$ when it is called by agent $j$ at time $t$ (0 otherwise). For $j \in {\cal M}_{-i}$, let $Y^{\alpha_i}_j(t)$ be a random variable that is equal to 1 when agent $i$ asks for recommendations from agent $j$ at time $t$ for its own user (0 otherwise). 
Let the (random) reward agent $i$ gets from the set of recommendations made by $(\boldsymbol{\alpha}, \boldsymbol{\beta})$ to its user at time $t$ be
\begin{align}
& O^i_{\boldsymbol{\alpha}, \boldsymbol{\beta}}(t, x_i(t)) := \sum_{f \in {\cal F}_i} Y^{\alpha_i}_{f}(t) p^i_{f} q_{f}(x_i(t), {\cal N}_i(t)) \notag \\
& + \hspace{-0.1in} \sum_{j \in {\cal M}_{-i}} Y^{\alpha_i}_{j}(t) 
\left(\sum_{f \in {\cal F}_{j}} Y^{\beta_{j}, \alpha_i}_{f}(t) c_{i,j}  q_{f}(x_i(t), {\cal N}_i(t))  \right), \label{eqn:thisreward}
\end{align}
and let
%
\begin{align*}
 S^i_{\boldsymbol{\alpha}, \boldsymbol{\beta}}(T)  
&:= E_{\boldsymbol{\alpha}, \boldsymbol{\beta}}\left[\sum_{t=1}^T O^i_{\boldsymbol{\alpha}, \boldsymbol{\beta}}(t, x_i(t))\right],
\end{align*}
%
be the total expected reward agent $i$ can get based only on recommendations to its own users by time $T$. Let 
\begin{align}
U^{i}_{\boldsymbol{\alpha}, \boldsymbol{\beta}}(t, x_j(t)) 
:= Y^{\alpha_j}_i(t)\sum_{f \in {\cal F}_i} Y^{\beta_i, \alpha_j}_f(t), \label{eqn:thiscom}
\end{align}
be the number of items of agent $i$ purchased by agent $j$'s user at time $t$. Clearly we have
\begin{align*}
Q^i_{\boldsymbol{\alpha}, \boldsymbol{\beta}}(T) 
=  S^i_{\boldsymbol{\alpha}, \boldsymbol{\beta}}(T)   + E_{\boldsymbol{\alpha}, \boldsymbol{\beta}}\left[\sum_{t=1}^T  \sum_{j \in {\cal M}_{-i}} c_{i,j} U^{i}_{\boldsymbol{\alpha}, \boldsymbol{\beta}}(t, x_j(t))   \right].
\end{align*}
We can see that $Q^i_{\boldsymbol{\alpha}, \boldsymbol{\beta}}(T) \geq S^i_{\boldsymbol{\alpha}, \boldsymbol{\beta}}(T)$. Agent $i$'s goal is to maximize its total reward $S^i_{\boldsymbol{\alpha}, \boldsymbol{\beta}}(T)$ from its own users for any $T$. \rev{Since agents are cooperative, agent $i$ also helps other agents $j \in {\cal M}_{-i}$ to maximize $S^j_{\boldsymbol{\alpha}, \boldsymbol{\beta}}(T)$ by recommending its items to them. Hence the total reward the agent gets, $Q^i_{\boldsymbol{\alpha}, \boldsymbol{\beta}}(T)$, is at least $S^i_{\boldsymbol{\alpha}, \boldsymbol{\beta}}(T)$.}

We assume that user arrivals to the agents are independent of each other. Therefore, agent $j$ will also benefit from agent $i$ if its item can be sold by agent $i$. In this paper, we develop distributed online learning algorithms $(\boldsymbol{\alpha}, \boldsymbol{\beta})$ for the agents in ${\cal M}$ such that the expected total reward $S^i_{\boldsymbol{\alpha}, \boldsymbol{\beta}}(T)$ for any agent $i \in {\cal M}$ is maximized. This corresponds to minimizing the regret, which is given for agent $i$ at time $T$ as
\begin{align}
%
%
& R^i_{\boldsymbol{\alpha}, \boldsymbol{\beta}}(T) := \sum_{t=1}^T \mu_{i,k^*_i(x_i(t))}(x_i(t))- S^i_{\boldsymbol{\alpha}, \boldsymbol{\beta}}(T). \label{eqn:regret}
\end{align}
Note that the regret is calculated with respect to the highest expected reward agent $i$ can obtain from its own users, but not the users of other agents. Therefore, agent $i$ does not act strategically to attract the users of other agents, such as by cutting its own prices or paying commissions even when an item is not sold to increase its chance of being recommended by another agent. We assume that agents are fully cooperative and follow the rules of the proposed algorithms.
We provide a comparison between cooperative and strategic agents in the following remark. 
\begin{remark}
In our cooperative algorithms, an agent $j$ when called by agent $i$ recommends a set of items that has the highest probability of being purchased by agent $i$'s user. 
This recommendation does not decrease the reward of agent $j$, compared to the case when agent $j$ does not cooperate with any other agent, since we assume that $p^j_{f} \geq c_{i,j}$ for all $f \in {\cal F}_j$.
%
%
However, when the commission is fixed, recommending the set of items with the highest purchase rate does not always maximize agent $j$'s reward. For example, agent $j$ may have another item which has a slightly smaller probability of being purchased by agent $i$'s user, but has a much higher price than the item which maximizes the purchase rate. Then, it is of interest to agent $j$ to recommend that item to agent $i$ rather than the item that maximizes the purchase rate. This problem can be avoided by charging a commission which is equal to a percentage of the sales price of the item, i.e., $c_{i,j}(p^j_{f}) = c_{i,j} p^j_{f}$ for some $0<c_{i,j}<1$ for $f \in {\cal F}_j$. Our theoretical results will hold for this case as well. 
\newc{
We will numerically compare the effects of a fixed commission and a commission which is equal to a percentage of the sales price on the learning rates and set of items recommended by the agents in Section \ref{sec:num}.}
\end{remark}

We will show that the regret of the algorithms proposed in this paper will be sublinear in time, i.e., $R^i_{\boldsymbol{\alpha}, \boldsymbol{\beta}}(T) = O(T^\phi)$, $\phi<1$, which means that the distributed learning scheme converges to the average reward of the best recommender strategy ${\cal N}^*_i(x)$ given in (\ref{eqn:bestonestep}) (or equivalently $k^*_i(x)$ given in (\ref{eqn:armoptimal})) for each $i \in {\cal M}$, $x \in {\cal X}$, i.e., $R^i_{\boldsymbol{\alpha}, \boldsymbol{\beta}}(T) /T \rightarrow 0$. Moreover, the regret also provides us with a bound on how fast our algorithm converges to the best recommender strategy.

\vspace{-0.1in}
\section{Contextual Partitioning Algorithms for Multiple Recommendations} \label{sec:mainalgo}

In this section we propose a series of distributed online learning algorithms called {\em Context Based Multiple Recommendations} (CBMR). We denote the part of CBMR which gives an action to agent $i$ at every time slot for its own user with $\alpha^{\textrm{CBMR}}_i$, and the part of CBMR which gives an action to agent $i$ at every time slot for the recommendation request of another agent with $\beta^{\textrm{CBMR}}_i$. When clear from the context, we will drop the superscript.
Basically, an agent using CBMR forms a partition of the context space $[0,1]^d$, depending on the final time $T$, consisting of $(m_T)^d$ sets where each set is a $d$-dimensional hypercube with dimensions $1/m_T \times 1/m_T \times \ldots \times 1/m_T$. The sets in this partition are indexed by  ${\cal I}_T = \{1,2,\ldots, (m_T)^d  \}$. We denote the set with index $l$ with $I_l$.
Agent $i$ learns about the rewards and {\em purchase rates} of its actions in each set in the partition independently from the other sets in the partition based on the context information of the users that arrived to agent $i$ and the users for which agent $i$ is recommended by another agent. Since users with similar contexts have similar purchase probabilities (Assumptions \ref{ass:dependent} and \ref{ass:independent}), it is expected that the optimal recommendations are similar for users located in the same set $I_l \in {\cal I}_T$. 
Since the best recommendations are learned independently for each set in ${\cal I}_T$, there is a tradeoff between the number of sets in ${\cal I}_T$ and the estimation of the best recommendations for contexts in each set in ${\cal I}_T$. 
We will show that in order to bound regret sublinearly over time, the parameter $m_T$ should be non-decreasing in $T$.

There are two versions of CBMR: one for group-dependent purchase probabilities given by Assumption \ref{ass:dependent}, which is called CBMR-d, and the other for independent purchase probabilities given by Assumption \ref{ass:independent}, which is called CBMR-ind. 
The difference between these two is that CBMR-d calculates the expected reward from each action in ${\cal L}_i$ for agent $i$ separately, while CBMR-ind forms the expected reward of each action in ${\cal L}_i$ based on the expected rewards of the items recommended in the chosen action.
\rev{We have
\begin{align*}
|{\cal L}_i| = {|{\cal F}_i| \choose N} + \sum_{n=0}^{N-1} {|{\cal F}_i| \choose n}  {N-n+M-2 \choose M-1},
\end{align*}
which grows combinatorially in $N$ and $M$ and polynomially in $|{\cal F}_i|$}. 
We will show that CBMR-ind provides much faster learning than CBMR-d when Assumption \ref{ass:independent} holds. We explain these algorithms in the following subsections.

\vspace{-0.1in}
\subsection{Definition of CBMR-d} \label{sec:contextdep}

The pseudocode of CBMR-d is given in Fig. \ref{fig:CBMR}.
At any time $t$ in which agent $i$ chooses an action $k_i \in {\cal L}_i$, for any agent $j \in {\cal M}_i(k_i)$, it should send a request to that agent for $m_j(k_i)$ recommendations along with the context of its user $x_i(t)$ and the recommendation vector $\boldsymbol{m}(k_i)$.
If the agents do not want to share the set of recommendations and prices they made for agent $i$'s user with agent $i$, they can use the regulator scheme proposed in Section \ref{sec:regulator}. 
%
%
%
CBMR-d can be in any of the following three phases at any time slot $t$: the
{\em training} phase in which agent $i$ chooses an action $k_i$ such that it trains another agent $j$ by asking it for recommendations and providing $i$'s user's context information $x$, 
so that agent $j$ will learn to recommend its set of $m_j(k_i)$ items with the highest {\em purchase rate} for a user with context $x$, the {\em exploration} phase in which agent $i$ forms accurate estimates of the expected reward of actions $k_i \in {\cal L}_i$ by selecting $k_i$ and observing the purchases, and the {\em exploitation} phase in which agent $i$ selects the action in ${\cal L}_i$ with the highest estimated reward to maximize its total reward. The pseudocodes of these phases are given in Fig. \ref{fig:mtrain}.
%

At each time $t$, agent $i$ first checks which set in the partition ${\cal I}_T$ context $x_i(t)$ belongs to. 
%
We separate the set of actions in ${\cal L}_i$ into two. Let 
\begin{align*}
\hat{{\cal L}}_i &= \left\{ (\boldsymbol{u}_i,\boldsymbol{n}_{i}) \in {\cal L}_i \textrm{ such that }    n_{u_i} =0, ~\forall u_i \in {\cal M}_{-i}  \right\},
\end{align*}
be the set of actions in which all recommendations are from ${\cal F}_i$, and 
%
$\tilde{{\cal L}}_i = {\cal L}_i - \hat{{\cal L}}_i$
%
be the set of actions in which at least one recommendation is from an agent in ${\cal M}_{-i}$. 

The training phase is required for actions in $\tilde{{\cal L}}_i$, while only exploration and exploitation phases are required for actions in $\hat{{\cal L}}_i$. 
When agent $i$ chooses an action $k_i \in \tilde{{\cal L}}_i$, the agents $j \in {\cal M}_{i}(k_i)$ recommend $m_j(k_i)$ items from ${\cal F}_j$ to agent $i$. Recall that ${\cal N}_i(t)$ denotes the set of $N$ items that is recommended to agent $i$'s user at time $t$, based on the actions chosen by agent $i$ and the recommendations chosen by other agents for agent $i$.
\comment{
Therefore, the reward agent $i$ gets at time $t$ is 
%
$Q_i(t) := S_i(t) + O_i(t)$,
%
where 
\begin{align*}
S_i(t) := \hspace{-0.2in} \sum_{f \in {\cal N}_i(t) - {\cal F}_i} \hspace{-0.2in} c_{i,j(f)} I(f \in F_i(t)) + \hspace{-0.2in} \sum_{f \in {\cal F}_i \cap {\cal N}_i(t) } \hspace{-0.2in} p^i_f I(f \in F_i(t)),
\end{align*}
and
\begin{align*}
O_i(t) := \sum_{j \in {\cal M}_{-i}} O_{i,j}(t),
\end{align*}
where 
\begin{align*}
O_{i,j}(t) := \sum_{f \in {\cal F}_i \cap {\cal N}_j(t)} (p^i_f - c_{j,i}) I(f \in F_j(t)),
\end{align*}
and $F_i(t)$ is the set of items purchased by agent $i$'s user at time $t$.
}
For $k_i \in {\cal L}_i$, let $N^i_{k_i,l}(t)$ be the number of times action $k_i$ is selected in response to a context arriving to the set $I_l \in {\cal I}_T$ by time $t$.
Since agent $i$ does not know ${\cal F}_j$ and the purchase rates for $j \in {\cal M}_{i}(k_i)$,
before forming estimates about the reward of action $k_i$, it needs to make sure that $j$ will almost always recommend its set of $m_j(k_i)$ items with the highest purchase rate for agent $i$'s user. 
Otherwise, agent $i$'s estimate about the reward of action $k_i$ might deviate a lot from the correct reward of action $k_i$ with a high probability, resulting in agent $i$ choosing suboptimal actions most of the time, hence resulting in a high regret.
This is why the training phase is needed for actions in $\tilde{{\cal L}}_i$. 

\begin{figure}[htb]
\fbox {
\begin{minipage}{0.95\columnwidth}
{\fontsize{10}{10}\selectfont
\flushleft{Context Based Multiple Recommendations for dependent purchase probabilities (CBMR-d for agent $i$):}
\begin{algorithmic}[1]
\STATE{Input: $D_1(t)$, $D_{2,k}(t)$, $k \in \tilde{{\cal L}_i}$, $D_3(t)$, $T$, $m_T$}
\STATE{Initialize: Partition $[0,1]^d$ into $(m_T)^d$ sets, indexed by the set ${\cal I}_T = \{ 1,2,\ldots, (m_T)^d\}$. }
\STATE{$N^i_{k,l}=0, \forall k \in {\cal L}_i, l \in {\cal I}_T$, $N^i_{1,k,l}=0$, $N^i_{2,k,l}=0$, $\forall k \in \tilde{{\cal L}}_{i}, l \in {\cal I}_T$.}
\STATE{$\boldsymbol{N}^i_l = \{ N^i_{k,l}\}_{k \in {\cal L}_i}$, $\boldsymbol{N}^i_{2,l} = \{ N^i_{2,k,l}\}_{k \in \tilde{{\cal L}}_i}$.}
\STATE{$\bar{r}^i_{k,l}=0$, $\bar{\pi}^i_{k,l}=0$, for all $k \in {\cal L}_i, I_l \in {\cal I}_T$.}
\STATE{$\boldsymbol{\bar{r}}^i_{l} = \{ \bar{r}^i_{k,l} \}_{l \in {\cal L}_i}$, $\boldsymbol{\bar{\pi}}^i_{l} = \{ \bar{\pi}^i_{k,l} \}_{l \in {\cal L}_i}$.}
\WHILE{$t \geq 1$}
\STATE{Algorithm $\alpha_i$ (Send recommendations to own users)}
\FOR{$l=1,\ldots,(m_T)^d$}
\IF{$x_i(t) \in I_l$}
\IF{$\exists k \in \hat{{\cal L}}_i \textrm{ such that } N^i_{k,l} \leq D_1(t)$}
\STATE{Run {\bf Explore}($k$, $N^i_{k,l}$, $\bar{r}^i_{k,l}$, $t$)}
\ELSIF{$\exists k \in \tilde{{\cal L}}_{i} \textrm{ such that } N^i_{1,k,l} \leq D_{2,k}(t)$}
\STATE{Run {\bf Train}($k$, $N^i_{1,k,l}$, $t$)}
\ELSIF{$\exists k \in \tilde{{\cal L}}_{i} \textrm{ such that } N^i_{2,k,l} \leq D_3(t)$}
\STATE{Run {\bf Explore}($k$, $N^i_{2,k,l}$, $\bar{r}^i_{k,l}$, $t$)}
\ELSE
\STATE{Run {\bf Exploit}($\boldsymbol{N}^i_l, \boldsymbol{N}^i_{2,l}$, $\boldsymbol{\bar{r}}^i_l$, $t$)}
\ENDIF
\ENDIF
\ENDFOR
\STATE{Algorithm $\beta_i$ (Send recommendations to other agents' users)}
\FOR{$j \in {\cal M}_{-i}$ such that $i \in {\cal M}_j(\alpha_j(t))$}
\STATE{Receive $x_j(t)$ and $\boldsymbol{m}(\alpha_j(t))$}
\FOR{$l=1,\ldots,(m_T)^d$}
\IF{$x_j(t) \in I_l$}
\IF{$\exists$ $k \in {\cal B}_i(\boldsymbol{m}(\alpha_j(t)))$ such that  $N^i_{k,l} \leq D_1(t)$}
\STATE{Run {\bf Explore2}($k$, $N^i_{k,l}$, $\bar{\pi}^i_{k,l}$, $j$, $t$)}
\ELSE
\STATE{Run {\bf Exploit2}($\boldsymbol{N}^i_l$, $\boldsymbol{\bar{\pi}}^i_{l}$, $\boldsymbol{m}(\alpha_j(t))$, $j$, $t$)}
\ENDIF
\ENDIF
\ENDFOR
\ENDFOR
\STATE{$t=t+1$}
\ENDWHILE
\end{algorithmic}
}
\end{minipage}
} \caption{Pseudocode of CBMR-d for agent $i$.} \label{fig:CBMR}
\vspace{-0.2in}
\end{figure}
\begin{figure}[htb]
\fbox {
\begin{minipage}{0.95\columnwidth}
{\fontsize{10}{10}\selectfont
{\bf Train}($k$, $n$, $t$):
\begin{algorithmic}[1]
\STATE{Select action $\alpha_i(t)=k$.}
\STATE{Receive reward $\tilde{r}_k(t) := O^i_{\boldsymbol{\alpha}, \boldsymbol{\beta}}(t, x_i(t))$ given in (\ref{eqn:thisreward}).}
\STATE{ $n++$.}
\end{algorithmic}
{\bf Explore}($k$, $n$, $r$, $t$):
\begin{algorithmic}[1]
\STATE{Select action $\alpha_i(t)=k$.}
\STATE{Receive reward $\tilde{r}_k(t) := O^i_{\boldsymbol{\alpha}, \boldsymbol{\beta}}(t, x_i(t))$ given in (\ref{eqn:thisreward}).}
\STATE{$r = (n r + \tilde{r}_k(t))/(n + 1)$.}
\STATE{$n++$.}
\end{algorithmic}
{\bf Exploit}($\boldsymbol{n}$, $\boldsymbol{r}$, $t$):
\begin{algorithmic}[1]
\STATE{Select action $\alpha_i(t) \in \argmax_{k \in {\cal L}_i} r_k$.}
\STATE{Receive reward $\tilde{r}_{\alpha_i(t)}(t) := O^i_{\boldsymbol{\alpha}, \boldsymbol{\beta}}(t, x_i(t))$ given in (\ref{eqn:thisreward}).}
\STATE{$r_{\alpha_i(t)} = (n_{\alpha_i(t)} r_{\alpha_i(t)} + \tilde{r}_{\alpha_i(t)}(t))/(n_{\alpha_i(t)} + 1)$.}
\STATE{$n_{\alpha_i(t)}++$.}
\end{algorithmic}
{\bf Explore2}($k$, $n$, $\pi$, $j$, $t$):
\begin{algorithmic}[1]
\STATE{Recommend set of items ${\cal F}_i(k)$ to agent $j$.}
\STATE{Receive purchase feedback $\tilde{\pi}_k(t) := U^{i}_{\boldsymbol{\alpha}, \boldsymbol{\beta}}(t, x_j(t))$ given in (\ref{eqn:thiscom}) and the commission.}
\STATE{$\pi = (n \pi + \tilde{\pi}_k(t))/(n + 1)$.}
\STATE{$n++$.}
\end{algorithmic}
{\bf Exploit2}($\boldsymbol{n}$, $\boldsymbol{\pi}$, $\boldsymbol{m}$, $j$, $t$):
\begin{algorithmic}[1]
\STATE{Select action $\beta_i(t) \in \argmax_{k \in {\cal B}_i(\boldsymbol{m})} \pi_k$.}
\STATE{Recommend set of items ${\cal F}_i(\beta_i(t))$ to agent $j$.}
\STATE{Receive purchase feedback $\tilde{\pi}_{\beta_i(t)}(t) := U^{i}_{\boldsymbol{\alpha}, \boldsymbol{\beta}}(t, x_j(t))$ given in (\ref{eqn:thiscom}) and the commission.}
\STATE{$\pi_{\beta_i(t)} = (n_{\beta_i(t)} \pi_{\beta_i(t)} + \tilde{\pi}_{\beta_i(t)}(t))/(n_{\beta_i(t)} + 1)$.}
\STATE{$n_{\beta_i(t)}++$.}
\end{algorithmic}
}
\end{minipage}
} \caption{Pseudocode of the training, exploration and exploitation modules for agent $i$.} \label{fig:mtrain}
\vspace{-0.2in}
\end{figure}

In order to separate training, exploration and exploitation phases, agent $i$ keeps two counters for actions $k_i \in \tilde{{\cal L}}_i$ for each $I_l \in {\cal I}_T$.
The first one, i.e., $N^i_{1,k_i,l}(t)$, counts the number of user arrivals to agent $i$ with context in set $I_l$ in the training phases of agent $i$ by time $t$, in which it had chosen action $k_i$.
The second one, i.e., $N^i_{2,k_i,l}(t)$, counts the number of user arrivals to agent $i$ with context in set $I_l$ in the exploration and the exploitation phases of agent $i$ by time $t$, in which it had chosen action $k_i$.
The observations made at the exploration and the exploitation phases are used to estimate the expected reward of agent $i$ from taking action $k_i$.
In addition to the counters $N^i_{k_i,l}(t)$, $N^i_{1,k_i,l}(t)$ and $N^i_{2,k_i,l}(t)$, agent $i$ also keeps {\em control functions}, $D_1(t)$, $D_{2,k_i}(t)$, $k_i \in \tilde{{\cal L}}_i$ and $D_{3}(t)$, which are used together with the counters determine when to train, explore or exploit. 
The control functions are deterministic and non-decreasing functions of $t$ that are related to the minimum number of observations of an action that is required so that the estimated reward of that action is sufficiently accurate to get a low regret in exploitations. 
We will specify the exact values of these functions later when we prove our regret results.

In addition to the counters and control functions mentioned above, for each $k_i \in {\cal L}_i$ and $I_l \in {\cal I}_T$, agent $i$ keeps a {\em sample mean reward} $\bar{r}^i_{k_i,l}(t)$ which is the sample mean of the rewards (sum of prices of sold items and commissions divided by number of times $k_i$ is selected except the training phase) agent $i$ obtained from its recommendations when it selected action $k_i$ in its exploration and exploitation phases for its own user with context in $I_l$ by time $t$.
Agent $i$ also keeps a {\em sample mean purchase rate} $\bar{\pi}^i_{k_i,l}(t)$ which is the sample mean of the number of agent $i$'s own items in ${\cal F}_i(k_i)$ sold to a user (either agent $i$'s own user or user of another agent) with context in set $I_l$ when agent $i$ selects action $k_i$ for that user.\footnote{We will explain how agent $i$ selects an action $k_i$ when agent $j$ requests items from $i$ when we describe $\beta^{\textrm{CMBR-d}}_i$.}
%
Note that when agent $i$ chooses action $k_i$, agent $j \in {\cal M}_{-i}$ has ${|{\cal F}_j| \choose m_j(k_i)}$ different sets of items to recommend to agent $i$. 
In order for agent $j$ to recommend to agent $i$ its best set of items that maximizes the commission agent $i$ will get, agent $j$ must have accurate {\em sample mean purchase rates} $\bar{\pi}^j_{k_j,l}(t)$ for its own actions $k_j \in {\cal B}_j(\boldsymbol{m}(k_i))$.
Therefore, when a user with context $x_i(t) \in I_l$ arrives at time $t$, agent $i$ checks if the following set is nonempty:
%
\begin{align*}
&{\cal S}_{i,l}(t) := \left\{ k_i \in \hat{{\cal L}}_i \textrm{ such that } N^i_{k_i,l}(t) \leq D_1(t)  \textrm{ or } k_i \in \tilde{{\cal L}}_i  \right. \\
& \left. \textrm{ such that } N^i_{1,k_i,l}(t) \leq D_{2,k_i}(t) \textrm{ or } N^i_{2,k_i,l}(t) \leq D_{3}(t)   \right\},
\end{align*} 
in order to make sure that it has accurate estimates of the expected rewards of each of its actions $k_i \in {\cal L}_i$, as well as to make sure that other agents have accurate estimates about the {\em purchase rates} of their own set of items for a user with context in set $I_l$.

For $k_i \in \tilde{{\cal L}}_i$, let ${\cal E}^i_{k_i,l}(t)$ be the set of rewards collected from selections of action $k_i$ at times $t' \in [t]$ when agent $i$'s user's context is in set $I_l$, \rev{and all the other agents in ${\cal M}_{i}(k_i)$ are trained sufficiently, i.e., $N^i_{1,k_i,l}(t') > D_{2,k_i}(t')$.} 
%
%
For $k_i \in \hat{{\cal L}}_i$, let ${\cal E}^i_{k_i,l}(t)$ be the set of rewards collected from action $k_i$ by time $t$.
If ${\cal S}_{i,l}(t) \neq \emptyset$, then agent $i$ trains or explores by randomly choosing an action $\alpha_i(t) \in {\cal S}_{i,l}(t)$. If ${\cal S}_{i,l}(t) = \emptyset$, this implies that all actions in ${\cal L}_i$ have been trained and explored sufficiently, so that agent $i$ exploits by choosing the action with the highest sample mean reward, i.e.,
\begin{align}
\alpha_i(t) \in \argmax_{k_i \in {\cal L}_i} \bar{r}^i_{k_i,l}(t). \label{eqn:maximizer}
\end{align}
%
We have
%
$\bar{r}^i_{k_i,l}(t) = (\sum_{r \in {\cal E}^i_{k_i,l}(t)} r)/|{\cal E}^i_{k_i,l}(t)|$.\footnote{Agent $i$ does not need to keep ${\cal E}^i_{k_i,l}(t)$ in its memory. Keeping and updating $\bar{r}^i_{k_i,l}(t)$ online, as new rewards are observed is enough.}
%
\rev{When there is more than one action which has the highest sample mean reward, one of them is randomly selected.} 

The other part of CBMR-d, i.e., $\beta_i$, gives agent $i$ the set of items to recommend to agent $j$ when agent $j$ takes an action $\alpha_j(t) = k_j \in {\cal L}_j$ for which $i \in {\cal M}_j(k_j)$, and sends to agent $i$ its user's context $x_j(t) \in I_l \in {\cal I}_T$ and the recommendation vector $\boldsymbol{m}(k_j)$. 
In order to recommend the set of items with the maximum {\em purchase rate} for the user of agent $j$, agent $i$ should learn the {\em purchase rate} of its own items for the recommendation vector $\boldsymbol{m}(k_j)$. 
Agent $i$ responds to agent $j$'s request in the following way.
If there is any $k_i \in {\cal B}(\boldsymbol{m}(k_j))$ for which the {\em purchase rate} is under-explored, i.e., $N^i_{k_i,l}(t) \leq D_{1}(t)$, then agent $i$ recommends to agent $j$ the set of items ${\cal F}_i(k_i)$.
Otherwise if all {\em purchase rates} of actions in ${\cal B}(\boldsymbol{m}(k_j))$ are explored sufficiently, i.e., $N^i_{k_i,l}(t) > D_{1}(t)$ for all $k_i \in {\cal B}(\boldsymbol{m}(k_j))$, then agent $i$ recommends to agent $j$ the set of its own items which maximizes the {\em sample mean purchase rate}, i.e., ${\cal F}_i(\hat{k}^*_i(k_j,l,t))$, where
\begin{align*}
\hat{k}^*_i(k_j,l,t) = \argmax_{k_i \in {\cal B}(\boldsymbol{m}(k_j))} \bar{\pi}^i_{k_i,l}(t).
\end{align*}
%


%
%
%
%
%
In the following subsection we prove an upper bound on the regret of CBMR-d.

\vspace{-0.1in}
\subsection{Analysis of the regret of CBMR-d}\label{sec:analCBMR_d}

For each $I_l \in {\cal I}_T$ and $k_i \in {\cal L}_i$, let 
%
$\overline{\mu}_{i,k_i,l} := \sup_{x \in I_l} \mu_{i,k_i}(x)$,
$\underline{\mu}_{i,k_i,l} := \inf_{x \in I_l} \mu_{i,k_i}(x)$,
$\overline{\pi}_{i,k_i,l} := \sup_{x \in I_l} \pi_{i,k_i}(x)$ and
$\underline{\pi}_{i,k_i,l} := \inf_{x \in I_l} \pi_{i,k_i}(x)$.
%
Let $x^*_l$ be the context at the center of the set $I_l$. We define the {\em optimal reward} action of agent $i$ for set $I_l$ as
\begin{align*}
k^*_i(l) := \argmax_{k_i \in {\cal L}_i} \mu_{i,k_i}(x^*_l),
\end{align*}
and the {\em optimal purchase rate} action of agent $j$ for agent $i$ given agent $i$ selects action $k_i$ for set $I_l$ as
\begin{align*}
k^*_j(k_i,l) := \argmax_{k_j \in {\cal B}_j(\boldsymbol{m}(k_i))} \pi_{j,k_j}(x^*_l).
\end{align*}
Let
\begin{align*}
{\cal K}^i_{\theta,a_1,l}(t) := \left\{ k_i \in {\cal L}_i : \underline{\mu}_{i,k_i^*(l),l} - \overline{\mu}_{i,k_i,l} > a_1 t^{\theta} \right\},
\end{align*}
be the set of {\em suboptimal reward} actions for agent $i$ at time $t$, and 
\begin{align*}
{\cal Y}^{j,k_i}_{\theta,a_1,l}(t) := \left\{ k_j \in {\cal B}_j(\boldsymbol{m}(k_i)) : \underline{\pi}_{i,k^*_j(k_i,l) ,l} - \overline{\pi}_{i,k_j,l} > a_1 t^{\theta} \right\},
\end{align*}
be the set of {\em suboptimal purchase rate} actions of agent $j$ for agent $i$, given agent $i$ chooses action $k_i$ at time $t$, where $\theta<0$, $a_1 > 0$.
The agents are not required to know the values of the parameters $\theta$ and $a_1$. They are only used in our analysis of the regret. First, we will give regret bounds that depend on values of $\theta$ and $a_1$, and then we will optimize over these values to find the best bound. 
Let $Y_R := \sup_{x \in {\cal X}} \mu_{i,k^*_i(x)}(x)$, which is the maximum expected loss agent $i$ can incur for a time slot in which it chooses a suboptimal action.

The regret given in (\ref{eqn:regret}) can be written as a sum of three components: 
\begin{align*}
R^i(T) = E[R^i_e(T)] + E[R^i_s(T)] + E[R^i_n(T)],
\end{align*}
where $R^i_e(T)$ is the regret due to training and explorations by time $T$, $R^i_s(T)$ is the regret due to suboptimal action selections in exploitations by time $T$, and $R^i_n(T)$ is the regret due to near optimal action ($k_i \in {\cal L}_i - {\cal K}^i_{\theta,a_1,l}(t)$) selections in exploitations by time $T$ of agent $i$, which are all random variables. In the following lemmas we will bound each of these terms separately. The following lemma bounds $E[R^i_e(T)]$. \newc{Due to space limitations we give the lemmas in this and the following subsections without proofs. The proofs can be found in our online appendix \cite{jstsp_online}.}
%
%
\begin{lemma} \label{lemma:explorations}
When CBMR-d is run with parameters $D_1(t) = D_3(t) = t^{z} \log t$,
%
$D_{2,k_i}(t) = \max_{j \in {\cal M}_{i}(k_i)}  {F_{\max} \choose m_j(k_i)}  t^{z} \log t$,
%
and $m_T = \left\lceil T^{\gamma} \right\rceil$,\footnote{For a number $r \in \mathbb{R}$, let $\lceil r  \rceil$ be the smallest integer that is greater than or equal to $r$.} where $0<z<1$ and $0<\gamma<1/d$, we have
%
\begin{align*}
E[R^i_e(T)]  
& \leq Y_R 2^d \left( |{\cal L}_i| + |\tilde{{\cal L}}_i| {F_{\max} \choose \left\lceil F_{\max}/2   \right\rceil} \right) T^{z+\gamma d} \log T \\
&+ Y_R 2^d |{\cal L}_i| T^{\gamma d}.
\end{align*}
\end{lemma}
\delete{
\begin{proof}
We sum over all exploration and training steps by time $T$. The contribution to the regret is at most $Y_R$ in each of these steps.
\end{proof}
}

\remove{
\begin{proof}
Since time step $t$ is a training or an exploration step if and only if ${\cal S}_{i,l}(t) \neq \emptyset$, up to time $T$, there can be at most   $\left\lceil T^{z} \log T \right\rceil$ exploration steps in which an arm in $k_i \in {\cal F}_i$ is selected by agent $i$, 
$\left\lceil F_{\max} T^{z} \log T \right\rceil$ training steps in which agent $i$ selects agent $j_i \in {\cal M}_{-i}$, $\left\lceil T^{z} \log T \right\rceil$ exploration steps in which agent $i$ selects agent $j_i \in {\cal M}_{-i}$. Result follows from summing these terms and the fact that $(m_T)^d \leq 2^d T^{\gamma d}$ for any $T \geq 1$.
\end{proof}
}

From Lemma \ref{lemma:explorations}, we see that the regret due to explorations is linear in the number of sets in partition ${\cal I}_T$, i.e., $(m_T)^d$, hence exponential in parameters $\gamma$ and $z$. We conclude that $z$ and $\gamma$ should be small to achieve sublinear regret in training and exploration phases.

For any $k_i \in {\cal L}_i$ and $I_l \in {\cal I}_T$, the sample mean $\bar{r}^i_{k_i,l}(t)$ represents a random variable which is the average of the independent samples in set ${\cal E}^i_{k_i,l}(t)$. 
Different from classical finite-armed bandit theory \cite{auer}, these samples are not identically distributed. 
In order to facilitate our analysis of the regret, we generate two different artificial i.i.d. processes to bound the probabilities related to $\bar{r}^i_{k_i,l}(t)$, $k_i \in {\cal L}_i$.
The first one is the {\em best} process in which rewards are generated according to a bounded i.i.d. process with expected reward $\overline{\mu}_{i,k_i,l}$, the other one is the {\em worst} process in which rewards are generated according to a bounded i.i.d. process with expected reward $\underline{\mu}_{i,k_i,l}$. 
Let $r^{\textrm{best}}_{k_i,l}(\tau)$ denote the sample mean of the $\tau$ samples from the best process and $r^{\textrm{worst}}_{k_i,l}(\tau)$ denote the sample mean of the $\tau$ samples from the worst process.\com{Isn't z less than 1? Check notation. Cem: Is there a notation overlap? I changed it to $\tau$.} We will bound the terms $E[R^i_n(T)]$ and $E[R^i_s(T)]$ by using these artificial processes along with the \newc{H\"{o}lder condition} given in Assumption \ref{ass:dependent}.
Details of this are given in \cite{jstsp_online}.
The following lemma bounds $E[R^i_s(T)]$. 
\begin{lemma} \label{lemma:suboptimal1}
When CBMR-d is run with parameters $D_1(t) = D_3(t) = t^{z} \log t$,
%
$D_{2,k_i}(t) = \max_{j \in {\cal M}_{i}(k_i)} {F_{\max} \choose m_j(k_i)} t^{z} \log t$,
%
and $m_T = \left\lceil T^{\gamma} \right\rceil$, where $0<z<1$ and $0<\gamma<1/d$, given that
%
$2 L Y_R ( \sqrt{d})^\alpha t^{- \gamma \alpha} + 2 (Y_R + 2) t^{-z/2} \leq a_1 t^\theta$, 
we have
%
\begin{align*}
E[R^i_s(T)] \leq   2^{d+1} Y_R |\hat{{\cal L}}_i| \beta_2 T^{\gamma d} 
+ 2^{d+2} Y_R |\tilde{{\cal L}}_i| \beta_2 T^{\gamma d + z/2}/z .
\end{align*}
\end{lemma}
\delete{
\begin{proof}
This proof is similar to the proof of Lemma 2 in \cite{cem2013deccontext}. The difference is that instead of bounding the probabilities that a suboptimal arm in ${\cal F}_i$ will be chosen or another agent $j \in {\cal M}_{-i}$ chooses a suboptimal arm in ${\cal F}_j$ when called by agent $i$, we bound the probabilities that a suboptimal action in $\hat{{\cal L}}_i$ is chosen and a suboptimal action in $\tilde{{\cal L}}_i$ is chosen. Another difference is that every time a suboptimal action is chosen, the one-step regret can be at most $Y_R$.
\end{proof}
}

\rev{From Lemma \ref{lemma:suboptimal1}, we see that the regret increases exponentially with parameters $\gamma$ and $z$, similar to the result of Lemma \ref{lemma:explorations}. These two lemmas suggest that $\gamma$ and $z$ should be as small as possible, given that the condition
$2 L Y_R ( \sqrt{d})^\alpha t^{- \gamma \alpha} + 2 (Y_R + 2) t^{-z/2} \leq a_1 t^\theta$ 
is satisfied. 
}

When agent $i$ chooses an action $k_i \in {\cal L}_i$ such that $j \in {\cal M}_i(k_i)$,
there is a positive probability that agent $j$ will choose a suboptimal set of $m_j(k_i)$ items to recommend to agent $i$'s user, i.e., it will choose a {\em suboptimal purchase rate action} for agent $i$'s action. 
Because if this, even if agent $i$ chooses a near optimal action $k_i \in {\cal L}_i - {\cal K}^i_{\theta,a_1,l}(t)$ it can still get a low reward.
We need to take this into account in order to bound $E[R^i_n(T)]$. 
\add{The following lemma gives the bound on $E[R^i_n(T)]$.}
\remove{
For $j_i \in {\cal M}_{-i}$, let $X^i_{j_i,l}(t)$ denote the random variable which is the number of times a suboptimal arm of agent $j_i$ is chosen in exploitation steps of agent $i$ in which agent $i$ calls agent $j_i$ when the context is in set $P_l$ by time $t$. The next lemma bounds the expected number of times a suboptimal arm is chosen when agent $i$ calls another agent.
\begin{lemma} \label{lemma:callother}
When CLUP is run with parameters $D_1(t) = t^{z} \log t$, $D_2(t) = F_{\max} t^{z} \log t$, $D_3(t) = t^{z} \log t$ and $m_T = \left\lceil T^{\gamma} \right\rceil$, where $0<z<1$ and $0<\gamma<1/d$, given that
%
$2 L( \sqrt{d})^\alpha t^{- \gamma \alpha} + 6 t^{-z/2} \leq a_1 t^\theta$, 
%
we have $E[X^i_{j_i,l}(t)] \leq 2 F_{\max} \beta_2$ for $j_i \in {\cal M}_{-i}$.

\end{lemma}
\remove{
\begin{proof}
Let $\Xi^i_{j_i,l}(t)$ be the event that a suboptimal arm $m \in {\cal F}_{j_i}$ is called by agent $j_i$, when it is called by agent $i$ for the $t$-th time in the exploitation phase of agent $i$. 
We have 
%
$X^i_{j_i,l}(t) = \sum_{t'=1}^{{\cal E}^i_{j_i,l}(t)} I(\Xi^i_{j_i,l}(t'))$,
%
and
\begin{align*}
&P \left( \Xi^i_{j_i,l}(t) \right) 
\leq \sum_{m \in {\cal L}^{j_i}_\theta} P \left( \bar{r}_{m,l}(t) \geq \bar{r}^{*j_i}_{l}(t) \right) \\
&\leq \sum_{m \in {\cal L}^{j_i}_\theta}
\left(  P \left( \bar{r}_{m,l}(t) \geq \overline{\mu}_{m,l} + H_t, {\cal W}^i_{l}(t) \right) \right. \\  
& \left. + P \left( \bar{r}^{*j_i}_{l}(t) \leq \underline{\mu}^{*j_i}_{l} - H_t, {\cal W}^i_{l}(t) \right) 
 + P \left( \bar{r}_{m,l}(t) \geq \bar{r}^{*j_i}_{l}(t), \right. \right. \\ 
& \left. \left. \bar{r}_{m,l}(t) < \overline{\mu}_{m,l} + H_t,
 \bar{r}^{*j_i}_{l}(t) > \underline{\mu}^{*j_i}_{l} - H_t ,
{\cal W}^i_{l}(t) \right)  \right),
\end{align*}
where $*j_i$ is the optimal arm in ${\cal F}_{j_i}$.
Let $H_t = 2 t^{-z/2}$. Similar to the proof of Lemma \ref{lemma:suboptimal1}, the last probability in the sum above is equal to zero while the first two probabilities are upper bounded by $e^{-2(H_t)^2 t^z \log t}$. This is due to the second phase of the exploration algorithm which requires at least $t^z \log t$ samples from the second exploration phase for all agents before the algorithm exploits any agent. Therefore, we have
%
$P \left( \Xi^i_{j_i,l}(t) \right) \leq \sum_{m \in {\cal L}^{j_i}_\theta} 2 e^{-2(H_t)^2 t^z \log t} \leq 2 |{\cal F}_{j_i}|/t^2$.
%
These together imply that 
%
$E[X^i_{j_i,l}(t)] \leq \sum_{t'=1}^{\infty} P(\Xi^i_{j_i,l}(t')) \leq 2 |{\cal F}_{j_i}| \sum_{t'=1}^\infty 1/t^2$.
\end{proof}
}

We will use Lemma \ref{lemma:callother} in the following lemma to bound $E[R_n(T)]$.
}
\begin{lemma} \label{lemma:nearoptimal}
When CBMR-d is run with parameters $D_1(t) = D_3(t) = t^{z} \log t$,
%
$D_{2,k_i}(t) = \max_{j \in {\cal M}_{i}(k_i)} {F_{\max} \choose m_j(k_i)}   t^{z}  \log t$,
%
and $m_T = \left\lceil T^{\gamma} \right\rceil$, where $0<z<1$ and $0<\gamma<1/d$, given that
%
$2 L Y_R ( \sqrt{d})^\alpha t^{- \gamma \alpha} + 2 (Y_R + 2) t^{-z/2} \leq a_1 t^\theta$, 
we have
\begin{align*}
E[R^i_n(T)] \leq (2 a_1 T^{1+\theta})/(1+\theta) + 2 Y_R {F_{\max} \choose \left\lceil F_{\max}/2   \right\rceil} \beta_2.
\end{align*}
\end{lemma}
\delete{
\begin{proof}
If a near optimal action in $\hat{{\cal L}}_i$ is chosen at time $t$, the contribution to the regret is at most $a_1 t^{\theta}$. If a near optimal action  $k_i \in \tilde{{\cal L}}_i$ is chosen at time $t$, and if some of the agents in ${\cal M}_i(k_i)$ choose their near optimal items to recommend to agent $i$, then the contribution to the regret is at most $2 a_1 t^{\theta}$.
Therefore, the total regret due to near optimal action selections in ${\cal L}_i$ by time $T$ is upper bounded by 
%
$2 a_1 \sum_{t=1}^T t^{\theta}  \leq (2 a_1 T^{1+\theta})/(1+\theta)$,
by using the result in Appendix A in \cite{cem2013deccontext}. 
Each time a near optimal action in $k_i \in \tilde{{\cal L}}_i$ is chosen in an exploitation step, there is a small probability that an item chosen by some agent $j \in {\cal M}_i(k_i)$ is a suboptimal one. Using a result similar to Lemma 3 in \cite{cem2013deccontext}, the expected number of times a suboptimal arm is chosen is bounded by $2 {F_{\max} \choose \left\lceil F_{\max}/2   \right\rceil} \beta_2$. Each time a suboptimal arm is chosen, the regret can be at most $Y_R$.
\end{proof}
}

\rev{From Lemma \ref{lemma:nearoptimal}, we see that the regret due to near optimal actions depends exponentially on $\theta$ which is related to the negative of $\gamma$ and $z$. Therefore, $\gamma$ and $z$ should be chosen as large as possible to minimize the regret due to near optimal actions.} 
Combining the above lemmas, we obtain the finite time regret bound for agents using CBMR-d, which is given in the following theorem.
\begin{theorem}\label{theorem:cos}
When CBMR-d is run with parameters $D_1(t) = D_3(t) = t^{2\alpha/(3\alpha+d)} \log t$,
%
$D_{2,k_i}(t) = \max_{j \in {\cal M}_{i}(k_i)} {F_{\max} \choose m_j(k_i)}  t^{2\alpha/(3\alpha+d)} \log t$,
%
and $m_T = \lceil T^{1/(3\alpha + d)} \rceil$, we have
\begin{align*}
R^i(T) &\leq T^{\frac{2\alpha+d}{3\alpha+d}}
\left( \frac{4 (Y_R (L d^{\alpha/2} +1)+1)}{(2\alpha+d)/(3\alpha+d)} + Y_R 2^d Z_i \log T \right) \\
&+ T^{\frac{\alpha+d}{3\alpha+d}} \frac{2^{d+2} Y_R |\tilde{{\cal L}}_i| \beta_2}{2\alpha/(3\alpha+d)} \\
&+ T^{\frac{d}{3\alpha+d}} 2^d Y_R (2  |\tilde{{\cal L}}_i|  \beta_2 
+ |{\cal L}_i|) + 2 Y_R {F_{\max} \choose \left\lceil F_{\max}/2   \right\rceil} \beta_2,
\end{align*}
i.e., $R_i(T) = O \left(Z_i T^{\frac{2\alpha+d}{3\alpha+d}} \right)$,
where $Z_i = |{\cal L}_i| + |\tilde{{\cal L}}_i| {F_{\max} \choose \left\lceil F_{\max}/2   \right\rceil} $.
\end{theorem}
\begin{proof}
The highest orders of regret come from explorations and near optimal arms, which are $O(T^{\gamma d + z})$ and $O(T^{1+\theta})$ respectively. We need to optimize them with respect to the constraint
%
$2 L Y_R ( \sqrt{d})^\alpha t^{- \gamma \alpha} + 2 (Y_R + 2) t^{-z/2} \leq a_1 t^\theta$, 
which is assumed in Lemmas \ref{lemma:suboptimal1} and \ref{lemma:nearoptimal}.
The values that minimize the regret for which this constraint holds are $\theta = -z/2$, $\gamma = z/(2\alpha)$ $a_1 = 2 Y_R ( Ld^{\alpha/2} + 1) + 4$ and $z = 2\alpha/(3\alpha+d)$. 
The result follows from summing the bounds in Lemmas \ref{lemma:explorations}, \ref{lemma:suboptimal1} and \ref{lemma:nearoptimal}. 
\end{proof}

\newc{
\begin{remark}
A uniform partition of the context space such as ${\cal I}_T$ may not be very efficient when user arrivals have contexts that are concentrated into some (unknown) regions of the context space. 
In such a case, it is better to explore and train these regions more frequently than regions with few context arrivals.
%
%
Algorithms that start with the entire context space as a single set and then adaptively partition the context space into smaller regions as more and more users arrive may achieve faster convergence rates, i.e., smaller regret, for the types of arrivals mentioned above. Such an algorithm that will work for $N=1$ is given in \cite{cem2013deccontext}.
\end{remark}
}

\newc{Theorem \ref{theorem:cos} indicates that agent $i$ can achieve sublinear regret with respect to the best ``oracle" recommendation strategy which knows the purchase probabilities of all the items in ${\cal F}$. However, the learning rate of CBMR-d can be slow when $M$, $N$ and ${\cal F}$ are large since the set of actions is combinatorial in these parameters. As a final remark, since regret is sublinear, the average reward of CBMR-d converges to the average reward of the best ``oracle" recommendation strategy, i.e., $\lim_{T \rightarrow \infty} R^i(T)/T = 0$.
}

\comment{
\subsection{Analysis of the regret under sales commission}

In this section we analyze the regret of CBMA, assuming that the purchase probability of an item is independent of the purchase probabilities of other recommended items. A discussion when purchase probability depends on the set of items recommended together is given in Section \ref{}

\subsection{Analysis of the regret under referral commission}
}

\vspace{-0.1in}
\subsection{Definition of CBMR-ind} \label{sec:contextindep}

In this subsection we describe the learning algorithm CBMR-ind. We assume that Assumption \ref{ass:independent} holds.
let ${\cal J}_{i,j} := \{ 1_j, 2_j, \ldots, N_j\}$ denote the set of the number of recommendations agent $i$ can request from agent $j$, where we use the subscript $j$ to denote that the recommendations are requested from agent $j$.
Let $\tilde{{\cal J}}_i  :=\cup_{j \in {\cal M}_{-i}} {\cal J}_{i,j}$, and ${\cal J}_i := {\cal F}_i \cup \tilde{{\cal J}}_i$ be the set of {\em arms} of agent $i$. 
We have $|{\cal J}_i| = |{\cal F}_i| + (M-1)N$.
We denote an arm of agent $i$ by index $u$. 
For arm $u \in \tilde{{\cal J}}_i$, let $j(u)$ be the agent that is called for item recommendations to agent $i$'s user, and let $n(u)$ be the number of requested items from agent $j(u)$.\footnote{$j(u)$, $u \in \tilde{{\cal J}}_i$ is different from $j(f)$, $f \in {\cal F}$, which denotes the agent that owns item $f$, and $n(u)$ is different from $n_j(k_i)$, which denotes the number of items agent $j$ should recommend when agent $i$ chooses action $k_i \in {\cal L}_i$.}
For $u \in {\cal F}_i$, $j(u)=i$ and $n(u)=1$.
In CBMR-ind, at each time $t$, agent $i$ chooses a set of arms such that the total number of item recommendations it makes to its user is $N$.
It is evident that choosing such a set of arms is equivalent to choosing an action $k_i \in {\cal L}_i$. 
Every action $k_i\in {\cal L}_i$ maps to a unique set of arms. 
Thus, for an action $k_i \in {\cal L}_i$, let ${\cal A}(k_i)$ be the set of arms corresponding to $k_i$.

Let ${\cal G}_{j,n}(x) \subset {\cal F}_j$ be the set of $n$ items in ${\cal F}_j$ with the highest purchase probabilities for a user with context $x$.
For an arm $u \in {\cal F}_i$, let its {\em purchase rate} for a user with context $x$ be $\nu_{i,u}(x):=q_u(x)$, and for an arm $u \in \tilde{{\cal J}_i}$, let it be $\nu_{i,u}(x) := \sum_{f \in {\cal G}_{j(u),n(u)}  } q_f(x)$.
Since Assumption \ref{ass:independent} holds, we have
\begin{align*}
\mu_{i, k_i}(x) = \sum_{u \in {\cal F}_i \cap {\cal A}(k_i)} p^i_u \nu_{i,u}(x) 
+ \sum_{u \in {\cal A}(k_i) - {\cal F}_i} c_{i,j(u)} \nu_{i,u}(x) .
\end{align*}
Different from CBMR-d, which estimates the reward of each action in ${\cal L}_i$ separately, CBMR-ind estimates the {\em purchase rates} of the arms in ${\cal J}_i$, and uses them to construct the estimated rewards of actions in ${\cal L}_i$.
%
The advantage of CBMR-ind is that the purchase rate estimate of an arm $u \in {\cal J}_i$ can be updated based on the purchase feedback whenever any action ${\cal L}_i$ that contains arm $u$ is selected by agent $i$. 
%

The pseudocode of CBMR-ind is given in Fig. \ref{fig:CBMRind}. CMBR-ind partitions the context space in the same way as CBMR-d.
Unlike CBMR-d, CBMR-ind does not have exploration, exploitation and training phases for actions $k_i \in {\cal L}_i$. Rather than that, it has exploration and exploitation phases for each arm $u \in {\cal F}_i$, and exploration, exploitation and training phases for each arm $u \in \tilde{{\cal J}_i}$. 
Since a combination of arms is selected at each time slot, selected arms can be in different phases.
For each $u \in {\cal J}_i$ and $I_l \in {\cal I}_T$, CBMR-ind keeps the {\em sample mean purchase rate} $\bar{\nu}^i_{u,l}(t)$, which is the sample mean of the number of purchased items corresponding to arm $u$, purchased by users with contexts in $I_l$ in exploration and exploitation phases of agent $i$ by time $t$.

Similar to the counters and control functions of CBMR-d, CBMR-ind also keeps counters and control functions for each $u \in {\cal J}_i$ and $I_l \in {\cal I}_T$. 
Basically, for $u \in {\cal J}_i$, $N^i_{u,l}(t)$ counts the number of times arm $u$ is selected by agent $i$ to make a recommendation to a user with context in $I_l$ by time $t$. Counters $N^i_{1,u,l}(t)$ and $N^i_{2,u,l}(t)$ are only kept for arms $u \in \tilde{{\cal J}}_i$. The former one counts the number of times arm $u$ is trained by agent $i$ for times its users had contexts in $I_l$ by time $t$, while the latter one counts the number of times arm $u$ is explored and exploited by agent $i$ for times its users had contexts in $I_l$ by time $t$.
Let 
\begin{align}
&{\cal S}^{\textrm{ind}}_{i,l}(t) := \left\{ u \in {\cal F}_i \textrm{ such that } N^i_{u,l}(t) \leq D_1(t)  \textrm{ or } u \in \tilde{{\cal J}}_i  \right. \notag \\
& \left. \textrm{ such that } N^i_{1,u,l}(t) \leq D_{2,u}(t) \textrm{ or } N^i_{2,u,l}(t) \leq D_{3}(t)  \right\}, \label{eqn:theset} 
\end{align} 
where $D_1(t)$, $D_{2,u}(t)$, $u \in \tilde{{\cal J}}_i$ and $D_3(t)$ are counters similar to the counters of CBMR-d. Assume that $x_i(t) \in I_l$.
If ${\cal S}^{\textrm{ind}}_{i,l}(t) \neq \emptyset$, then agent $i$ randomly chooses an action $\alpha_i(t) \in {\cal L}_i$ from the set of actions for which ${\cal A}(k_i) \cap {\cal S}^{\textrm{ind}}_{i,l}(t) \neq \emptyset$.
Else if ${\cal S}^{\textrm{ind}}_{i,l}(t) =\emptyset$, then agent $i$ will choose the action
\begin{align}
\alpha_i(t) = \argmax_{k_i \in {\cal L}_i} \sum_{u \in {\cal F}_i \cap {\cal A}(k_i)} 
\hspace{-0.2in} p^i_u \bar{\nu}^i_{u,l}(t)
+ \hspace{-0.2in}  \sum_{u \in {\cal A}(k_i) - {\cal F}_i}  \hspace{-0.2in}  c_{i,j(u)} \bar{\nu}^i_{u,l}(t). \label{eqn:indepmaxim}
\end{align}
After action $\alpha_i(t)$ is chosen, the counters and sample mean purchase rates of arms in ${\cal A}(\alpha_i(t))$ are updated based on in which phase they are in.

Agent $i$ responds to item requests of other agents with users with contexts in $I_l \in {\cal I}_T$ in the following way. Any under-explored item $u \in {\cal F}_i$, i.e., $N^i_{u,l}(t) \leq D_1(t)$ is given priority to be recommended. If the number of under-explored items is less than the number of requested items, then the remaining items are selected from the set of items in $u \in {\cal F}_i$ with the highest sample mean purchase rates such that $N^i_{u,l}(t) > D_1(t)$. The pseudocode of this is not given due to limited space.

\begin{figure}[htb]
\fbox {
\begin{minipage}{0.95\columnwidth}
{\fontsize{10}{10}\selectfont
\flushleft{Context Based Multiple Recommendations for independent purchase probabilities (CBMR-ind for agent $i$):}
\begin{algorithmic}[1]
\STATE{Input: $D_1(t)$, $D_{2,u}(t)$, $u \in \tilde{{\cal J}}_{i}$, $D_3(t)$, $T$, $m_T$.}
\STATE{Initialize: Partition $[0,1]^d$ into $(m_T)^d$ sets, indexed by the set ${\cal I}_T = \{ 1,2,\ldots, (m_T)^d\}$.}
\STATE{$N^i_{u,l}=0, \forall u \in {\cal J}_i, I_l \in {\cal I}_T$, $N^i_{1,u,l}=0, N^i_{2,u,l}=0, \forall u \in \tilde{{\cal J}}_{i}, l \in {\cal I}_T$.}
\STATE{$\bar{\nu}^i_{u,l}=0, \forall u \in {\cal J}_i, I_l \in {\cal I}_T$.}
\WHILE{$t \geq 1$}
\FOR{$l=1,\ldots,(m_T)^d$}
\IF{$x_i(t) \in I_l$}
\IF{${\cal S}^{\textrm{ind}}_{i,l}(t) \neq \emptyset$ (given in (\ref{eqn:theset}))}
\STATE{Choose $\alpha_i(t)$ randomly from the subset of ${\cal L}_i$ such that ${\cal A}(\alpha_i(t)) \cap {\cal S}^{\textrm{ind}}_{i,l}(t) \neq \emptyset$.}
\ELSE
\STATE{Choose $\alpha_i(t)$ as in (\ref{eqn:indepmaxim}).}
\ENDIF
\ENDIF
\ENDFOR
\STATE{Receive reward $O^i_{\boldsymbol{\alpha}, \boldsymbol{\beta}}(t, x_i(t))$ given in (\ref{eqn:thisreward}).}
\FOR{$u \in {\cal A}(\alpha_i(t))$}
\IF{$u \in \tilde{{\cal J}}_i$ and $N^i_{1,u,l} \leq D_{2,u}(t)$}
\STATE{$N^i_{1,u,l}++$.}
\ELSIF{$u \in \tilde{{\cal J}}_i$ and $N^i_{1,u,l} > D_{2,u}(t)$}
\STATE{$\bar{\nu}^i_{u,l} = \frac{N^i_{2,u,l} \bar{\nu}^i_{u,l} + U^{j(u)}_{\boldsymbol{\alpha}, \boldsymbol{\beta}}(t, x_{i}(t))}{N^i_{2,u,l} +1}$ (given in (\ref{eqn:thiscom})).}
\STATE{$N^i_{2,u,l}++$.}
\ELSE
\STATE{$\bar{\nu}^i_{u,l} = (N^i_{u,l} \bar{\nu}^i_{u,l} + Y^{\alpha_i}_u(t))/(N^i_{2,u,l} +1)$.}
\STATE{$N^i_{u,l}++$.}
\ENDIF
\ENDFOR
\STATE{$t=t+1$.}
\ENDWHILE
\end{algorithmic}
}
\end{minipage}
} \caption{Pseudocode of CBMR-ind for agent $i$.} \label{fig:CBMRind}
\vspace{-0.2in}
\end{figure}

\comment{
\begin{figure}[htb]
\fbox {
\begin{minipage}{0.95\columnwidth}
{\fontsize{10}{10}\selectfont
{\bf Choose}($N$, ${\cal N}$, $\boldsymbol{r}$):
\begin{algorithmic}[1]
\STATE{Select arms  $u \in {\cal J}_i - {\cal N}$ such that $u \notin  {\cal J}_{i,j}$ if $\exists$ $u' \in {\cal N} \cap {\cal J}_{i,j}$ for $j \in {\cal M}_{-i}$, $\sum_{u} n_u \leq N$ and $\sum_{u} r_u$ is maximized.}
\end{algorithmic}
{\bf Play}(${\cal N}$, ${\cal N}_1$, ${\cal N}_2$, ${\cal N}_3$ $\boldsymbol{N}$, $\boldsymbol{r}$):
\begin{algorithmic}[1]
\STATE{Take action ${\cal N}$, get the recommendations of other agents, recommend ${\cal N}_i(t)$ to the user.}
\FOR{$u \in {\cal N}$}
\IF{$u \in {\cal N}_i(t) \cap {\cal F}_i$}
\STATE{Receive reward $r_u(t) = I(u \in F_i(t))$. $r_u = \frac{N_{u,l} r_u + r_u(t) }{N_{u,l} + 1}$, $N_{u,l}++$.}
\ELSIF{$u \in ({\cal N} - {\cal F}_i) \cap {\cal N}^t_i$}
\STATE{Receive reward $r_u(t) = \sum_{f \in {\cal F}_j} I(f \in F_i(t))$, $N_{1,u,l}++$ }
\ELSE
\STATE{Receive reward $r_u(t) = \sum_{f \in {\cal F}_j} I(f \in F_i(t))$, $r_u = \frac{N_{2,u,l} r_u + r_u(t) }{N_{2,u,l} + 1}$, $N_{2,u,l}++$ $N_{2,u,l}++$}
\ENDIF
\ENDFOR
\end{algorithmic}
}
\end{minipage}
} \caption{Pseudocode of the choose and play modules for agent $i$.} \label{fig:mtrain2}
\vspace{-0.3in}
\end{figure}
}

\vspace{-0.1in}
\subsection{Analysis of the regret of CBMR-ind}\label{sec:CBMRindan}

In this subsection we bound the regret of CBMR-ind. Under Assumption \ref{ass:independent}, the expected reward of an item $f$ to agent $i$ for a user with context $x$ is $\kappa_{i,f}(x) := p^i_f q_f(x)$ for $f \in {\cal F}_i$ and $\kappa_{i,f}(x) := c_{i,j} q_f(x)$ for $f \in {\cal F}_j$. 
%
%
%
For a set of items ${\cal N}$, let $f_n({\cal N})$ denote the item in ${\cal N}$ with the $n$th highest expected reward for agent $i$. 
For an item $f \in {\cal F}$ and $I_l \in {\cal I}_T$, let $\underline{\kappa}_{i,f,l} := \inf_{x \in I_l} \kappa_{i,f}(x)$, and $\bar{\kappa}_{i,f,l} := \sup_{x \in I_l} \kappa_{i,f}(x)$.
For the set $I_l$ of the partition ${\cal I}_T$, the set of suboptimal arms of agent $i$ at time $t$ is given by
\begin{align*}
{\cal U}^i_{\theta,a_1,l}(t) &:= \left\{ u \in {\cal F}_i: \underline{\kappa}_{i, f_{N}({\cal F}), l}  - \bar{\kappa}_{i, u, l} \geq a_1 t^{\theta} \right\} \\
& \cup
\left\{  u \in \tilde{{\cal J}}_i :
\underline{\kappa}_{i, f_{N}({\cal F}), l} - \bar{\kappa}_{i, f_{n(u)}({\cal F}_{j(u)}),l} \geq a_1 t^{\theta} \right\}.
\end{align*}
We will optimize over $a_1$ and $\theta$ as we did in Section \ref{sec:analCBMR_d}.
The set of near-optimal arms at time $t$ for $I_l$ is ${\cal J}_i - {\cal U}^i_{\theta,a_1,l}(t)$.
Similar to the approach we took in Section \ref{sec:analCBMR_d}, we divide the regret into three parts: $R^i_e(T)$, $R^i_s(T)$ and $R^i_n(T)$, and bound them individually. 
Different from the analysis of CBMR-d, here $R^i_s(T)$ denotes the regret in time slots in which all selected arms are exploited and at least one them is suboptimal, while $R^i_n(T)$ denotes the regret in time slots in which all selected arms are exploited and all of them are near-optimal. 
In the following lemma, we bound the regret of CBMR-ind due to explorations and trainings. 
 
\begin{lemma}\label{lemma:explorationsind}
When CBMR-ind is run by the agents with $D_1(t) = t^{z} \log t$,
%
$D_{2,u}(t) =  {F_{\max} \choose n(u)} t^{z} \log t$,
%
$D_3(t) = t^{z} \log t$ and $m_T = \left\lceil T^{\gamma} \right\rceil$,
where $0<z<1$ and $0<\gamma<1/d$, we have
\add{\vspace{-0.1in}}
\begin{align*}
& E[R^i_e(T)] \leq
Y_R 2^d (|{\cal J}_i| + (M-1) N) T^{\gamma d} \\
&+Y_R 2^d \left( |{\cal J}_i| + (M-1) \sum_{z=1}^N {F_{\max} \choose z} \right) T^{z+\gamma d} \log T.
\end{align*}
\com{$(m_T)^d ({\cal J}_i + (M-1) N)$ comes from the fact that we have to explore/train when we have the number of explorations/trainings is equal to the control function.}
\end{lemma}
\delete{
\begin{proof}
For a set $I_l \in {\cal I}_T$, the regret due to explorations is bounded by $|{\cal J}_i| \left\lceil T^z \log T \right\rceil$. Agent $i$ spends at most $\sum_{z=1}^N {F_{\max} \choose z} \left\lceil T^z \log T \right\rceil$ time steps to train agent $j$. Note that this is the worst-case number of trainings for which agent $j$ does not learn about the purchase probabilities of its items in set $I_l$ from its own users, and from the users of agents other than agent $i$. 
The result follows from summing over all sets in ${\cal I}_T$.
\end{proof}
}
\comment{\com{The regret can be much better than this. For example when arm $1_j$ is trained, the regret due to $D_{2,1_j}$ is at most $F_{\max} T^z \log T$. I think we can write the regret due to trainings as 
%
(M-1) \sum_{z=1}^N {F_{\max} \choose z} T^z \log T.
\end{align*}
}}

\rev{In the next lemma, we bound $E[R^i_s(T)]$.}

\begin{lemma} \label{lemma:suboptimal1ind}
When CBMR-ind is run by the agents with $D_1(t) = t^{z} \log t$,
%
$D_{2,u}(t) = {F_{\max} \choose n(u)} t^{z} \log t$,
%
$D_3(t) = t^{z} \log t$ and $m_T = \left\lceil T^{\gamma} \right\rceil$, where $0<z<1$ and $0<\gamma<1/d$, given that
%
$2 L Y_R ( \sqrt{d})^\alpha t^{- \gamma \alpha} + 2 (Y_R + 2) t^{-z/2} \leq a_1 t^\theta$, 
we have
%
\begin{align*}
E[R^i_s(T)] \leq   2^{d+1} Y_R |{\cal J}_i| \beta_2 T^{\gamma d} 
+ 2^{d+2} N Y_R |\tilde{\cal J}_i| \beta_2 T^{\gamma d + z/2}/z .
\end{align*}
\end{lemma}
\delete{
\begin{proof}
Let $\Omega$ denote the space of all possible outcomes, and let $w$ be a sample path. Let ${\cal W}^i_{l}(t) := \{ w : S_{i,l}(t) = \emptyset \}$ denote the event that CBMR-ind is in the exploitation phase at time $t$.
The idea is to bound the probability that agent $i$ selects at least one
suboptimal arm in an exploitation step, and then using this to bound the expected number of times a suboptimal arm is selected by agent $i$.
Let ${\cal V}^i_{u,l}(t)$ be the event that a suboptimal action $u \in {\cal J}_i$ is chosen at time $t$ by agent $i$. We have
%
$E[R^i_s(T)] \leq Y_R \sum_{l \in {\cal I}_T} \sum_{t=1}^T \sum_{u \in {\cal U}^i_{\theta,l}(t)} P({\cal V}^i_{u,l}(t), {\cal W}^i_{l}(t) )$.

Similar to the proof of Lemma 2 in \cite{cem2013deccontext}, by using a Chernoff bound it can be shown that for any $u \in {\cal U}^i_{\theta,l}(t) \cap {\cal F}_i$, we have $P({\cal V}^i_{u,l}(t), {\cal W}^i_{l}(t) ) \leq 2/t^2$, and for any $u \in {\cal U}^i_{\theta,l}(t) \cap \tilde{{\cal J}}_i$, we have $P({\cal V}^i_{u,l}(t), {\cal W}^i_{l}(t) ) \leq 2/t^2 + 2 N |{\cal F}_{u(j)}| \beta_2/ t^{1-z/2}$.
Here different from Lemma 2 in \cite{cem2013deccontext}, $N$ comes from a union bound which is required to bound the probability that agent $j$ will recommend an item which is not in the set of the best $n(u)$ items of agent $j$ when agent $i$ chooses arm $u$.
\newc{We get the final regret result by summing the $P({\cal V}^i_{u,l}(t), {\cal W}^i_{l}(t) )$s over the set of suboptimal arms in each set in the partition ${\cal I}_T$, over the sets in the partition ${\cal I}_T$, and over time, and then by using the result in Appendix A in \cite{cem2013deccontext}.} 
\end{proof}
}

Different from Lemma \ref{lemma:suboptimal1}, $E[R^i_s(T)]$ is linear in $|{\cal J}_i|$ instead of in $|{\cal L}_i|$.
In the next lemma, we bound the $E[R^i_n(T)]$.

\begin{lemma} \label{lemma:nearoptimalind}
When CBMR-ind is run by the agents with $D_1(t) = t^{z} \log t$,
%
$D_{2,u}(t) = {F_{\max} \choose n(u)} t^{z} \log t$,
%
$D_3(t) = t^{z} \log t$ and $m_T = \left\lceil T^{\gamma} \right\rceil$, where $0<z<1$ and $0<\gamma<1/d$, given that
%
$2 L Y_R ( \sqrt{d})^\alpha t^{- \gamma \alpha} + 2 (Y_R + 2) t^{-z/2} \leq a_1 t^\theta$,
we have
\begin{align*}
E[R^i_n(T)] \leq (2 N a_1 T^{1+\theta})/(1+\theta) + 2 Y_R N F_{\max} \beta_2.
\end{align*}
\end{lemma}
\delete{
\begin{proof}
Since agent $i$ can choose at most $N$ arms at each time step, 
if all arms chosen at time $t$ by agent $i$ and by the other agents called by agent $i$ are near-optimal, then the contribution to the regret is bounded by $2 N a_1 t^{\theta}$. This is because for a near optimal arm $u \in {\cal F}_i$, the contribution to the regret is at most $a_1 t^{\theta}$, while for a near optimal arm in $u \in \tilde{{\cal J}}_i$, if agent $j(u)$ chooses its near optimal items to recommend to agent $i$, the contribution to the regret is at most $2 a_1 t^{\theta}$.
Therefore, the total regret due to near optimal arm selections by agent $i$ and the agents $i$ calls by time $T$ is upper bounded by 
$2 N a_1 \sum_{t=1}^T t^{\theta}  \leq (2 N a_1 T^{1+\theta})/(1+\theta)$
from the result in Appendix A in \cite{cem2013deccontext}. 

However, each time a near optimal arm in $u \in \tilde{{\cal J}}_i$ is chosen in an exploitation step, there is a small probability that an item chosen recommended by agent $j(u)$ is a suboptimal one. Using a result similar to Lemma 3 in \cite{cem2013deccontext}, the expected number of times a suboptimal arm is chosen by the called agent is bounded by $ n(u) F_{\max} \beta_2$. Each time a suboptimal item is chosen by the called agent $j(u)$, the regret can be at most $Y_R$. The result follows from summing these terms.
\com{Why not the expected number of times a suboptimal arm is chosen by the called agent is bounded by ${F_{\max} \choose n(u)} \beta_2$? Since agent $j$ wants to choose $n(u)$ arms with the highest purchase probabilities, by using a union bound it can be shown that probability that at least one of the chosen arms is suboptimal is bounded by $n(u) F_{\max} \beta_2$.}
\end{proof}
}

Combining the above lemmas, we obtain the regret bound for agents using CBMR-ind.
\begin{theorem}\label{theorem:cosind}
When CBMR-ind is run by the agents with $D_1(t) = t^{2\alpha/(3\alpha+d)} \log t$,
%
$D_{2,u}(t) = { F_{\max} \choose n(u) } t^{2\alpha/(3\alpha+d)} \log t$,
%
$D_3(t) = t^{2\alpha/(3\alpha+d)} \log t$ and $m_T = \lceil T^{1/(3\alpha + d)} \rceil$, we have
\begin{align*}
& R^i(T) \leq T^{\frac{2\alpha+d}{3\alpha+d}}
\left( \frac{4 N (Y_R (L d^{\alpha/2} +1)+1)}{(2\alpha+d)/(3\alpha+d)} + Y_R 2^d Z'_i \log T \right) \\
&+ T^{\frac{\alpha+d}{3\alpha+d}} \frac{2^{d+2} N Y_R |\tilde{{\cal J}}_i| \beta_2}{2\alpha/(3\alpha+d)} \\
&+ T^{\frac{d}{3\alpha+d}} 2^{d} Y_R (2 |{\cal J}_i|  \beta_2 + |{\cal J}_i| + (M-1)N  )
 + 2 Y_R N F_{\max} \beta_2,
\end{align*}
i.e., $R^i(T) = O \left( Z'_i T^{\frac{2\alpha+d}{3\alpha+d}} \right)$,
where $Z'_i = |{\cal J}_i| + (M-1) \sum_{z=1}^N {F_{\max} \choose z}$.
\end{theorem}
\begin{proof}
The highest orders of regret come from $E[R^i_e(T)]$ and $E[R^i_n(T)]$, which are $O(T^{\gamma d + z})$ and $O(T^{1+\theta})$, respectively. We need to optimize them with respect to the constraint
%
$2 L Y_R ( \sqrt{d})^\alpha t^{- \gamma \alpha} + 2 (Y_R + 2) t^{-z/2} \leq a_1 t^\theta$, 
which is assumed in Lemmas \ref{lemma:suboptimal1ind} and \ref{lemma:nearoptimalind}.
This gives us $\theta = -z/2$, $\gamma = z/(2\alpha)$ $a_1 = 2 Y_R (Ld^{\alpha/2} +1) + 4$ and $z = 2\alpha/(3\alpha+d)$. 
The result follows from summing the bounds in Lemmas \ref{lemma:explorationsind}, \ref{lemma:suboptimal1ind} and \ref{lemma:nearoptimalind}. 
\end{proof}

\newc{The result of Theorem \ref{theorem:cosind} indicates that the regret of CMBR-ind is sublinear in time and has the same time order as the regret of CBMR-d. However, the regret of CMBR-ind depends linearly on $|{\cal J}_i|$, which is much better than the linear dependence of the regret of CBMR-d on $|{\cal L}_i|$.}

We would also like to note that the number of trainings and explorations, and hence the regret, can be significantly reduced when agent $i$ knows $|{\cal F}_j|$ for all $j \in {\cal M}_{-i}$. 
In this case, agent $i$ can use the control function $D_{2,u}(t) := {|{\cal F}_{j(u)}| \choose n(u)} t^{2\alpha/(3\alpha+d)} \log t$ to decide if it needs to train arm $u$.

\vspace{-0.1in}
\subsection{Comparison of CBMR-d and CBMR-ind}

In this subsection we compare CBMR-d and CBMR-ind in terms of their regret bounds, training and exploration rates, and memory requirements. 
Note that the regret bound of CBMR-d depends on the size of the action space ${\cal L}_i$, which grows combinatorially with $M$ and $N$, and is an $N$ degree polynomial of $|{\cal F}_i|$. In contrast the size of the arm space ${\cal J}_i$ is just  $|{\cal F}_i| + (M-1)N$.
This is due to the fact that CBMR-d explores and exploits each action without exploiting the correlations between different actions. When the purchase probabilities depend on the set of items offered together, in the worst case there may be no correlation between rewards of different actions, and therefore the best one can do is to form independent sample mean estimates for each action in ${\cal K}_i$. 
\newc{
However, when the purchase probabilities are independent of the items offered together, since the expected reward of agent $i$ from an action $k_i \in {\cal L}_i$ is the sum of the expected rewards of the individual arms chosen by agent $i$ in that action, substantial improvement over the regret bound is possible due to smaller number of explorations and trainings.
}
%

Another advantage of CBMR-ind is that it requires a significantly smaller amount of memory than CBMR-d. CBMR-d needs to keep sample mean rewards and sample mean purchase rates for all actions and for all partitions of the context space, while CBMR-ind only needs to keep sample mean purchase rates for all arms and for all partitions of the context space.
\section{Performance as a Function of the Network} \label{sec:discuss}

In our analysis in the previous section we assumed that all agents are directly connected to each other.
In reality some agents may not be connected to each other. For example, agents $i$ and $j$ can both be connected to agent $j'$, but there may not exist a direct link between agent $i$ and agent $j$. This can happen when, for example, a pair of companies has a trade agreement between each other, but there is no trade agreement between companies $i$ and $j$. We assume that whenever a trade link exists between agents $i$ and $j$ it is bidirectional. 
In this case, even though $i$ cannot ask $j$ for items to recommend, $i$ can ask $j'$ and $j'$ can ask $j$ for an item. Then, if $i$ sells $j$'s item, agent $i$ will get commission $c_{i,j'}$ from $j'$, while agent $j'$ will get commission $c_{j', j} + c_{i,j'}$ from agent $j$ so that it recovers the payment it made to agent $i$ from agent $j$. We call agent $j'$ the {\em relay agent}.
Let ${\cal M}^{\textrm{dir}}_i$ denote the set of agents that are directly connected to agent $i$.

If agent $i$ only cooperates with the agents in ${\cal M}^{\textrm{dir}}_i$ and all cooperating agents use CBMR-d or CBMR-ind, then agent $i$ can achieve the sublinear regret bounds given in Theorems \ref{theorem:cos} and \ref{theorem:cosind}, with respect to the optimal distributed recommendation policy involving only the agents in $\{i\} \cup {\cal M}^{\textrm{dir}}_i$.
However, since agent $i$ cannot exploit the advantage of the items of other agents which are in ${\cal M}_{-i} - {\cal M}^{\textrm{dir}}_i$, its regret can be linear with respect to the optimal distributed recommendation policy involving all the agents.

CBMR-d and CBMR-ind can be modified in the following way to account for agents distributed over the network. Let ${\cal M}_{i,j}(p)$ be the set of relay agents between agents $i$ and $j$ when they are connected through path $p$. Let $D_p := |{\cal M}_{i,j}(p)|$. Let $h_{1}(p)$ be the agent that is connected directly to agent $i$ in path $p$, let $h_2(p)$ be the agent that is connected directly to agent $h_1(p)$, and so on. Then, the commission framework is modified such that when agent $i$ sells agent $j$'s item it gets a commission $c_{i,h_1(p)}$ from agent $h_1(p)$, agent $h_1(p)$ gets a commission $c_{h_1(p), h_2(p)} + c_{i,h_1(p)}$ from agent $h_2(p)$, and so on, such that 
%
$c_{h_{D_p}(p), j} \leq p^j_{f}$,
%
where $f \in {\cal F}_j$ is the item recommended by agent $j$ to agent $i$'s user.
Using this scheme, all agents benefit from the series of transactions described above, thus it is better for them to cooperate based on the rules of the commission scheme than to act on their own.
We assume that agent $j$ will not recommend an item to agent $h_{D_p}(p)$ if $c_{h_{D_p}(p), j} > p^j_{f}$ for all $f \in {\cal F}_j$.
%
Assume a connected network of agents $G({\cal M}, E)$ in which the maximum degree is $D_k$ and the longest path has $D_h$ hops, where $E$ denotes the set of direct links between the agents. Assume that the commissions $c_{i,j}>0$ are given between any agent $i$ and $j$ with direct links. 
\newc{We define agent $i$'s regret $R^i_{G({\cal M},E)}(T)$ to be the difference between the expected total reward of the optimal policy for which agent $i$ has access to the set of all the items of all agents in the network $G({\cal M},E)$ (but if the item is in ${\cal F}_j$, it gets the commission from agent $j'$, which is the agent that is directly connected to agent $i$ in the lowest commission path from agent $i$ to agent $j$) and the expected total reward of the learning algorithm used by agent $i$.}
The exploration and training control functions are run using $(D_k)^{D_h} F_{\max}$ instead of $F_{\max}$. This way agent $i$ will help the relay agents learn the other agent's recommendations accurately such that sublinear regret bounds can be achieved. The following theorem gives a bound on the regret of agents when they use the modified version of CBMR-ind discussed above, which we call CBMR-ind-N. A similar bound can also be proven for the modified version of CMBR-d.

\begin{theorem}\label{theorem:cosnetwork}
When Assumption \ref{ass:independent} holds, if CBMR-ind-N is run by all agents in $G({\cal M},E)$, with agent $i$ running CBMR-ind-N with the set of arms ${\cal J}_i$  as described in Section \ref{sec:contextindep}, defined using ${\cal M}^{\textrm{dir}}_i$ and ${\cal F}_i$ instead of ${\cal M}$ and ${\cal F}_i$,
control functions $D_1(t) = D_3(t) = t^{2\alpha/(3\alpha+d)} \log t$,
%
$D_{2,u}(t) = { (D_k)^{D_h} F_{\max} \choose n(u) } t^{2\alpha/(3\alpha+d)} \log t$,
%
and $m_T = \lceil T^{1/(3\alpha + d)} \rceil$; \newc{and if  commissions are such that all agents $j \in {\cal M}_{-i}$ will get a positive reward when their items are sold by agent $i$},\footnote{If the commission an agent needs to pay to sell an item to the user of agent $i$ is greater than the price of all the items of that agent, then that agent can be removed from the set of agents which agent $i$ can cooperate with. Then, our results will hold for the remaining set of agents.} we have,
\begin{align*}
& R^i_{G({\cal M},E)}(T) \leq  T^{\frac{\alpha+d}{3\alpha+d}} \frac{2^{d+2} N Y_R |\tilde{{\cal J}}_i| \beta_2}{2\alpha/(3\alpha+d)} \\
&+   2 Y_R N (D_k)^{D_h} F_{\max} \beta_2  \\
&+ T^{\frac{2\alpha+d}{3\alpha+d}}
\left( \frac{4 N (Y_R (L d^{\alpha/2} +1)+1)}{(2\alpha+d)/(3\alpha+d)} 
+ Y_R 2^d Z'_i \log T \right) \\ 
&+ T^{\frac{d}{3\alpha+d}} 2^{d} Y_R (2 |{\cal J}_i|  \beta_2 + |{\cal J}_i| + (M-1)N  ),
\end{align*}
i.e., $R^i_{G({\cal M},E)}(T) = O \left( Z'_i T^{\frac{2\alpha+d}{3\alpha+d}} \right)$,
where $Z'_i := |{\cal J}_i| + (M-1) \sum_{z=1}^N {(D_k)^{D_h} F_{\max} \choose z}$.
\end{theorem}%
\begin{proof}
The proof is similar to the proof of Theorem \ref{theorem:cosind}, but more explorations and trainings are required to ensure that all agents in all paths learn the rewards of their arms accurately for each set in ${\cal I}_T$, before exploiting them.
\end{proof}

Theorem \ref{theorem:cosnetwork} indicates that the regret increases exponentially in $D_h$ and is an $D_h N$th degree polynomial in $D_k$. While this makes CBMR-ind-N impractical in non-small world networks, CBMR-ind-N can achieve small regret in small world networks in which most agents are not directly connected to each other but all agents can be reached a with small number of hops.
\newc{Because of the additional commission the relay agent gets, an item which $j$ will recommend when it is directly connected to agent $i$ may not be recommended by $j$ when it is connected via agent $j'$. This results in sub-optimality compared to the case when all agents are connected. In Section \ref{sec:num} we numerically compare the effect of the agents' commissions on the performance.}

More refined versions of Theorem \ref{theorem:cosnetwork} can be derived if we focus on specific networks. One interesting network structure is when there is a monopolist agent which is directly connected to all of the other agents, while other agents are only directly connected to the monopolist agent. Corollary \ref{cor:monopolist} gives the regret bound for agent $i$ when it is the monopolist agent, and Corollary \ref{cor:notmonopolist} gives the regret bound for agent $i$ when it is not the monopolist agent.

\begin{corollary}\label{cor:monopolist}
When agent $i$ is the monopolist agent, if it runs CBMR-ind-N with $D_{2,u}(t) = {F_{\max} \choose n(u)} t^{2\alpha/(3\alpha+d)} \log t$ and everything else remaining the same as in Theorem \ref{theorem:cosnetwork}, it will achieve $R^i_{G_1({\cal M},E)}(T) = O \left( ({\cal F}_i + (M-1) N) T^{\frac{2\alpha+d}{3\alpha+d}} \right)$.
\end{corollary}

\begin{corollary}\label{cor:notmonopolist}
When agent $i$ is a non-monopolist agent, if it runs CBMR-ind-N with $D_{2,u}(t) = { M^2 F_{\max} \choose n(u)} t^{2\alpha/(3\alpha+d)} \log t$ and everything else remaining the same as in Theorem \ref{theorem:cosnetwork}, it will achieve $R^i_{G_2({\cal M},E)}(T) = O \left( ({\cal F}_i + N) M^{2N} T^{\frac{2\alpha+d}{3\alpha+d}} \right)$.
\end{corollary}

\comment{
\begin{figure}
\begin{center}
\includegraphics[width=0.3\columnwidth]{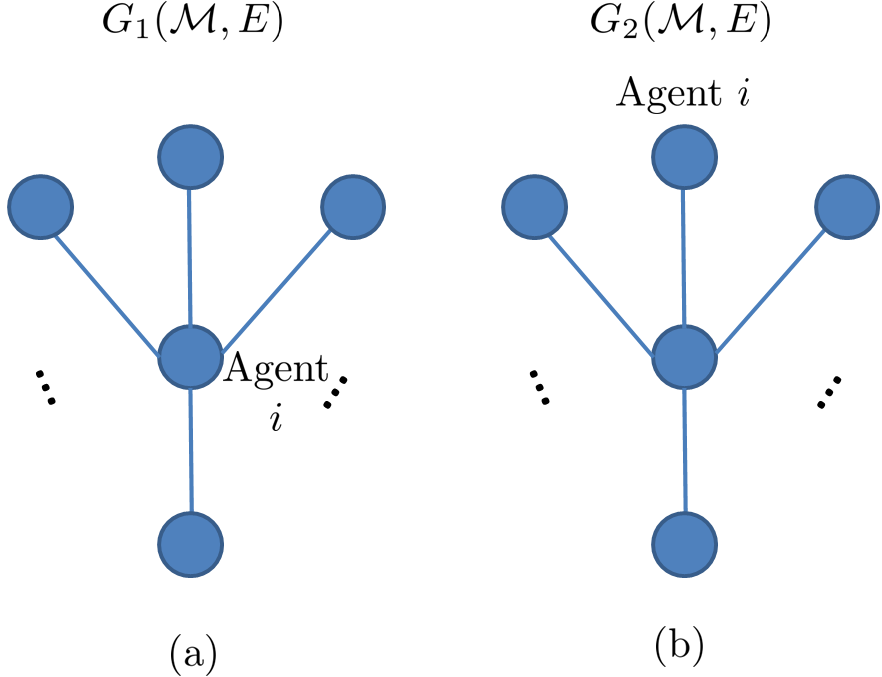}
\caption{Different types of network connectivity from the perspective of agent 1.} 
\label{networkcon2}
\end{center}
\end{figure}
}

Corollaries \ref{cor:monopolist} and \ref{cor:notmonopolist} imply that the learning is much faster, and hence the regret is much smaller when an agent has direct connections with more agents. This is because agent $i$ learns directly about the purchase probabilities of items in agent $j$'s inventory when it is directly connected to it, while the learning is indirect through a relay agent otherwise.

\vspace{-0.1in}
\section{Numerical Results} \label{sec:num}


\subsection{Description of the Data Set}

The Amazon product co-purchasing network data set includes product IDs, sales ranks of the products, and for each product the IDs of products that are frequently purchased with that product. This data is collected by crawling the Amazon website \cite{leskovec2007dynamics} and contains $410,236$ products and $3,356,824$ edges between products that are frequently co-purchased together. 
We simulate CBMR-ind and CBMR-d using the following distributed data sets adapted based on Amazon data. For a set of $N_1$ chosen products from the Amazon data set, we take these products and other $F_1$ products that are frequently co-purchased with the set of $N_1$ products as our set of items.

The set of products that are taken in the first step of the above procedure is denoted by ${\cal C}_h$. The set of all products, i.e., ${\cal F}$, contains these $N_1$ products and the set of products frequently co-purchased  with them, which we denote by ${\cal C}_f$.
We assume that each item has a unit price of $1$, but have different purchase probabilities for different types of users. 
Since user information is not present in the data set, we generate it artificially by assuming that every incoming user searches for a specific item in ${\cal C}_h$. This search query (item) will then be the context information of the user, hence the context space is ${\cal C}_h$. Thus, we set ${\cal I}_T = {\cal C}_h$.
Based on this, the agent that the user arrives to recommends $N$ items to the user.
The agent's goal is to maximize the total number of items sold to its users.

\newc{We generate the purchase probabilities in the following way: For group-dependent purchase probabilities, when $n$ of the products recommended for context $x \in {\cal C}_h$ are in the set of frequently co-purchased products with item $x$, then the purchase probability of each of these products will be $g_c(n) = 1 - an$, where $0<a<1/N$.}
%
%
For the other $N-n$ products which are not frequently co-purchased with item $x$, their purchase probability is $g_{nc} = b$, where $0<b<1-a$. 
For independent purchase probabilities, when a product recommended for context $x$ is in the set of frequently co-purchased products with item $x$, the purchase probability of that product will be $g_c$. When it is not, the purchase probability of that product will be $g_{nc}$, for which we have $g_c > g_{nc}$.

We assume that there are 3 agents and evaluate the performance of agent 1 based on the number of users arriving to agent 1 with a specific context $x^*$, which we take as the first item in set ${\cal C}_h$.
We assume that $T = 100,000$, which means that $100,000$ users with context $x^*$ arrive to agent 1. Since the arrival rate of context $x^*$ can be different for the agents, we assume arrivals with context $x^*$ to other agents are drawn from a random process. We take $N_1 = 20$, $F_1 =2$ and $N=2$. As a result, we get $30$ distinct items in ${\cal F}$ which are distributed among the agents such that $|{\cal F}_i|=10$ for every agent $i$.
Since the context space is discrete we have $d=1$, and there is no H\"{o}lder condition on the purchase probabilities as given in Assumptions \ref{ass:dependent} and \ref{ass:independent}, hence we take $\alpha=1/13$ such that $2\alpha/(3\alpha+d) = 1/8$. Unless otherwise stated, we assume that $g_c = 0.1$, $g_{nc} =0.01$, $a=0.5$ and $c_{i,j} = 0.5$.

%
%

\vspace{-0.1in}
\subsection{Comparison of the reward and regret of CBMR-d and CBMR-ind for group-dependent purchase probabilities}

We run both CBMR-d and CBMR-ind for group-dependent purchase probabilities assuming that both items that are frequently co-purchased with context $x^*$ are in agent 1's inventory. The ``oracle" optimal policy recommends one of the frequently co-purchased items and another item in agent 1's inventory to the user with context $x^*$, instead of recommending the two frequently co-purchased items together. Expected total reward of the ``oracle" optimal policy, total rewards of CBMR-d and CBMR-ind, and the number of trainings of CBMR-d and CBMR-ind are shown in Table \ref{tab:comparetwo} for agent 1. We have $|{\cal L}_1| = 58$ and $|{\cal J}_1| = 14$.

We see that the total reward of CBMR-ind is higher than the total reward of CBMR-d for this case. This is due to the fact that CBMR-d spends more than double the time CBMR-ind spends in training and exploration phases since it trains and explores each action in ${\cal L}_1$ separately. The time averaged regrets of CBMR-d and CBMR-ind are given in Fig. \ref{fig:comparetwo}. It is observed that CBMR-ind performs better than CBMR-d in all time slots, and the time averaged regret goes to $0$. From these results it seems that CBMR-ind is a good alternative to CBMR-d even for group-dependent purchase probabilities.

\begin{table}[b]
\vspace{-0.2in}
\centering
{\fontsize{10}{10}\selectfont
\setlength{\tabcolsep}{.3em}
\begin{tabular}{|l|c|c|c|}
\hline
 & Optimal & CBMR-d & CBMR-ind  \\
\hline
Total Reward & 11000 & 8724 & 9485\\
\hline
Trainings & - & 12391 & 5342 \\
\hline
\end{tabular}
}
\caption{Total rewards of the ``oracle" optimal policy, CBMR-d and CBMR-ind, and number of trainings of CBMR-d and CBMR-ind for group-dependent purchase probabilities.}
\label{tab:comparetwo}
\vspace{-0.1in}
\end{table}

\comment{
\begin{figure}
\centering
\parbox{7cm}{
	\includegraphics[width=7cm]{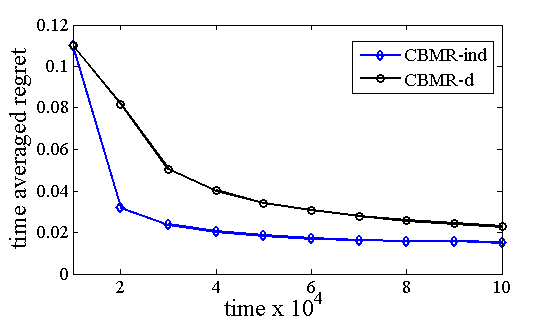}
	\vspace{-0.2in}
\caption{Time averaged regrets of CBMR-d and CBMR-ind for dependent purchase probabilities.} 
\label{fig:comparetwo}}%
\qquad
\begin{minipage}{7cm}
\includegraphics[width=7cm]{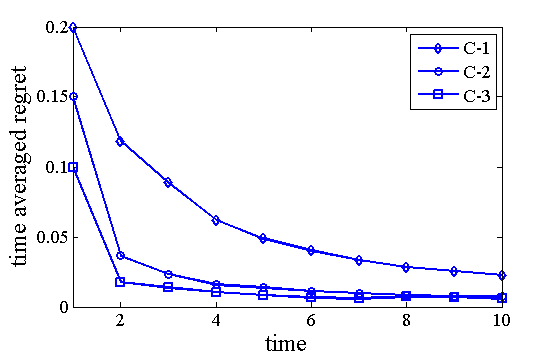}
\caption{Time averaged regret of CBMR-ind for independent purchase probabilities when agent 1 has both frequently co-purchased items (C-1), only one of the frequently co-purchased items (C-2) and none of the frequently co-purchased items (C-3).} 
\label{fig:setitems}
\end{minipage}
\vspace{-0.in}
\end{figure}
}
 
\begin{figure}[b]
\begin{center}
\includegraphics[width=0.9\columnwidth]{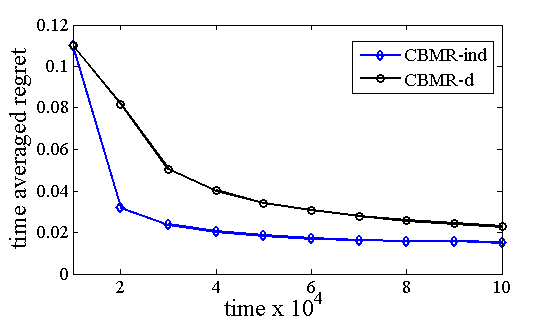}
\vspace{-0.1in}
\caption{Time averaged regrets of CBMR-d and CBMR-ind for group-dependent purchase probabilities.} 
\vspace{-0.2in}
\label{fig:comparetwo}
\end{center}
\end{figure}

\vspace{-0.1in}
\subsection{Effect of commission on the performance}

In this subsection we numerically simulate the effect of commissions $c_{1,j}$ that agent $1$ charges to other agents on the total reward of agent $1$. We consider CBMR-ind for independent purchase probabilities.
We assume that agent 1 has one of the frequently co-purchased items for context $x^*$, while agent 3 has the other frequently co-purchased item. The total reward of agent 1 as a function of the commissions $c_{1,2} = c_{1,3} =c$ is given in Table \ref{tab:commission}. We note that there is no significant difference in the total reward when the commission is  from $0$ to $0.1$. This is because $0.1$ commission is not enough to incentivize agent $1$ to recommend other agent's items to its users. 
However, for commissions greater than $0.1$, the optimal policy recommends the two frequently co-purchased items together, hence agent 1 learns that it should get recommendations from agent 3. Therefore, after $0.1$ the total reward of the agent is increasing in the commission. 
Another remark is that $c=0$ corresponds to the case when agent 1 is not connected to the other agents. Therefore, this example also illustrates how the rewards change as network connectivity changes. Since prices of items are set to 1 in this section, the commission agent 1 charges can increase up to 1. But if the prices are different, then the commission cannot exceed the recommended item's price. 
\newc{
In theory, in order to maximize its total reward from its own users, agent $i$ can adaptively select its commissions $c_{i,j}$, $j \in {\cal M}_{-i}$ based on what it learned about the purchase probabilities.  CBMR-d and CBMR-ind can be modified to adaptively select the commissions. Due to the limited space, we leave this as a future research topic.
%
}

\begin{table}[t]
\centering
{\fontsize{10}{10}\selectfont
\setlength{\tabcolsep}{.3em}
\begin{tabular}{|l|c|c|c|c|c|c|}
\hline
Commission $c$ & 0 & 0.1 & 0.2 & 0.3 & 0.4 & 0.5 \\
\hline
Reward  & 10471 & 10422 & 11476 & 12393 & 13340 & 14249  \\
(CBMR-ind) & & & & & & \\
\hline
\end{tabular}
}
\caption{Total reward of agent 1 as a function of the commission it charges to other agents.}
\label{tab:commission}
\vspace{-0.2in}
\end{table}

\vspace{-0.1in}
\subsection{Effect of the set of items of each agent on the performance}

In this subsection we compare three cases for independent purchase probabilities when agents use CBMR-ind. In C-1 agent $1$ has both items that are frequently co-purchased in context $x^*$, in C-2 it has one of the items that is frequently co-purchased in context $x^*$, and in C-3 it has none of the items that are frequently co-purchased in context $x^*$. The total reward of agent 1 for these cases is 17744, 14249 and 9402 respectively, while the total expected reward of the optimal policy is 20000, 15000 and 10000 respectively. Note that the total reward for C-3 is almost half of the total reward for C-1 since the commission agent 1 gets for a frequently co-purchased item is $0.5$. The time averaged regret of CBMR-ind for all these cases is given in Figure \ref{fig:setitems}. We see that the convergence rate for C-1 is slower than C-2 and C-3. This is due to the fact that in all of the training slots in C-1 a suboptimal set of items is recommended, while for C-2 and C-3 in some of the training slots the optimal set of items is recommended.

\begin{figure}[ht]
\begin{center}
\includegraphics[width=0.9\columnwidth]{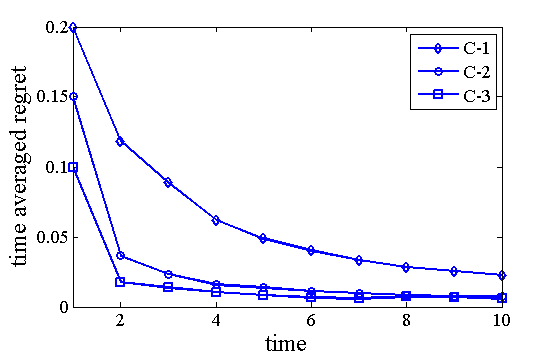}
\vspace{-0.1in}
\caption{Time averaged regret of CBMR-ind for independent purchase probabilities when agent 1 has both frequently co-purchased items (C-1), only one of the frequently co-purchased items (C-2) and none of the frequently co-purchased items (C-3).} 
\vspace{-0.3in}
\label{fig:setitems}
\end{center}
\end{figure}

\comment{
\subsection{Effect of network connectivity on the performance}

In this subsection we consider the network topologies given in Fig. \ref{fig:networkcon}, for independent purchase probabilities using CBMR-ind. We assume that the items are randomly distributed to the agents such that their inventory sizes are almost the same. In Fig. \ref{fig:networkcon} (a), agent $1$ has individual agreements with all other agents, while in Fig. \ref{fig:networkcon} (b), agent $1$ only has an agreement with agent $j$, but agent $j$ can provide exchanges between $1$ and other agents, however agent $1$'s commission will be lower than what it gets from direct exchanges.\com{confusing sentence}

We assume that agent $1$ can only request an item from an agent it is connected to. If it sells an item of an agent which it is connected to via agent $j$, then agent $i$ gets commission $c_j = 0.05$, while if it sells an item of an agent which it is directly connected to it gets commission $c_i = 0.2$. In the first case agent $j$ also takes commission from the agent that recommends its item to agent $i$, and thus agent $i$'s commission is small. We compare the total reward agent $i$ gets under these two cases in Table \ref{}. 
}

\vspace{-0.2in}
\section{Conclusion} \label{sec:conc}
In this paper we have presented a novel set of algorithms for multi-agent learning within a decentralized social network, and characterized the effect of different network structures on performance. Our algorithms are able to achieve sublinear regret in all cases, with the regret being much smaller if the user's acceptance probabilities for different items are independent. This paper can be used as a groundwork for agents in many different types of networks to cooperate in a mutually beneficial manner, from companies, charities, and celebrities who wish to generate revenue to artists, musicians, and photographers who simply want to have their work publicized. By cooperating in a decentralized manner and using our algorithms, agents have the benefit of retaining their autonomy and privacy while still achieving near optimal performance. 
\vspace{-0.1in}
\bibliographystyle{IEEE}
\bibliography{OSA}

\end{document}